\documentclass[11pt]{article}

\usepackage{amsmath}
\usepackage{enumitem}
\setlist[description]{style=nextline, labelindent=0pt}
\usepackage{amssymb}
\usepackage{amsfonts}
\usepackage{textcomp}
\usepackage{algorithm}
\usepackage{authblk}
\usepackage[noend]{algpseudocode}
\usepackage{amsthm}
\usepackage{xspace}
\usepackage{fullpage}
\usepackage[english]{babel}
\usepackage[utf8]{inputenc}
\let\originalS\S 
\usepackage[funcfont=roman]{complexity}
\renewcommand{\S}{\originalS}
\usepackage[pagebackref=true, colorlinks=true, citecolor=DarkGreen, linkcolor=Navy]{hyperref}
\usepackage{mathtools}
\usepackage[svgnames]{xcolor}
\usepackage{bussproofs}
\usepackage{subfig}
\usepackage{todonotes}
\usepackage{thmtools}
\usepackage{thm-restate}
\usepackage[capitalise, noabbrev]{cleveref}
\usepackage{bm}
\usepackage[theorems,skins]{tcolorbox}

\usepackage[T1]{fontenc}
\usepackage{lmodern}

\usepackage{etoolbox}

\patchcmd{\thebibliography}
  {\labelsep}
  {\labelsep \itemsep=0pt \parsep=0pt \parskip=0pt}
  {}{}

\newtcolorbox{mymathbox}[1][]{colback=white, colbacktitle = white, coltitle = black, #1}

\usepackage[%
activate={true,nocompatibility},
selected=true,
tracking=true,
factor=1100,
babel=true
]%
{microtype}

\usepackage
[%
labelformat=default, %
labelsep=colon, %
textformat=simple, %
font={small,sf}, %
labelfont=bf %
]%
{caption}

\makeatletter
 \renewenvironment{abstract}{%
   \small%
   \begin{center}%
     {\bfseries %
     \abstractname\vspace{-.5em}\vspace{\z@}}%
   \end{center}%
   \quotation}
 {\endquotation}
\makeatother

\theoremstyle{plain}
\newtheorem{theorem}{Theorem}[section]
\newtheorem{lemma}[theorem]{Lemma}
\newtheorem{corollary}[theorem]{Corollary}

\newtheorem{claim}[theorem]{Claim}
\newtheorem{observation}[theorem]{Observation}

\newtheorem*{integrality-assumption}{Lemma~\ref{lem:ipr-alekhnovich-razborov}}

\newtheorem*{reducibility-gives-reduction-operator}{Lemma~\ref{lem:reducibility-gives-reduction-operator}}

\theoremstyle{definition}
\newtheorem{remark}[theorem]{Remark}

\newtheorem{definition}[theorem]{Definition}

\newtheorem{fact}[theorem]{Fact}

\usepackage{todonotes}

\newcommand{\newauthor}[3]{
    \newcounter{#1comment}
    \setcounter{#1comment}{1}
    \expandafter\newcommand\csname #1\endcsname[1]{%
            \todo[size = \footnotesize, backgroundcolor = {#3}, caption = {}]{
                \arabic{#1comment}:
                {##1} --~\textbf{#2}
            }
            \addtocounter{#1comment}{1}
    }
    \expandafter\newcommand\csname #1changed\endcsname[1]{%
        \ifdraft
            \colorbox{#3}{
                ##1
            }%
        \else
            ##1%
        \fi
    }
}

\newauthor{jc}{Jonas}{orange}
\newauthor{shp}{Shuo}{yellow}
\newauthor{yg}{Yassine}{LimeGreen}

\newcommand{\set}[1]{\{ #1 \}}

\newcommand{\intersection}{\cap}

\newcommand{\union}{\cup}

\newcommand{\eqperiod}{\enspace .}
\newcommand{\eqcomma}{\enspace ,}

\newcommand{\floor}[1]{\lfloor #1 \rfloor}

\newcommand{\N}{\mathbb{N}}
\newcommand{\Z}{\mathbb{Z}}
\newcommand{\F}{\mathbb{F}}

\renewcommand{\epsilon}{\varepsilon}

\newcommand{\gnd}{\mathbb{G}_{n, \graphdeg}}

\newcommand{\hypgraph}{H}
\newcommand{\hypedgeset}{\mathcal{E}}
\newcommand{\hypedgesubset}{\mathcal{F}}
\newcommand{\hypedgebadset}{\mathcal{X}}
\newcommand{\hypedgebadsetalt}{\mathcal{Y}}

\newcommand{\nullcons}{\ell}

\newcommand{\pseudred}{R}
\newcommand{\unsatp}{\mathcal{P}}
\newcommand{\pcdeg}{D}
\newcommand{\nvars}{n}
\newcommand{\polyhomo}[2]{\unsatp_{#1 \to #2}}
\newcommand{\verta}{a}
\newcommand{\homomphi}{\varphi}
\newcommand{\homompsi}{\psi}
\newcommand{\homomu}{\mu}
\newcommand{\verti}{i}
\newcommand{\ala}{A}
\newcommand{\alt}{T}
\newcommand{\kcons}{k}
\newcommand{\ideal}[1]{\langle #1 \rangle}
\newcommand{\redop}{R} 

\newcommand{\ar}{\operatorname{ar}}
\newcommand{\relsymbol}{\Gamma}

\newcommand{\verts}[1]{\operatorname{Verts}_{#1}}

\newcommand{\pcpolyp}{p}
\newcommand{\pcpolyq}{q}
\newcommand{\indeterminate}{x}

\newcommand{\proofstd}{\pi}
\newcommand{\sop}{S}
\newcommand{\uniform}{t}
\newcommand{\sparseparam}{a}
\newcommand{\hoplength}{\tau}
\newcommand{\Desc}[1]{\operatorname{Desc}(#1)}
\newcommand{\clsize}{s}

\newcommand{\ring}{K}
\newcommand{\degstd}{D}
\newcommand{\graphdeg}{d}
\newcommand{\colcons}{c}

\newcommand{\col}[1]{\operatorname{Col}(#1)}

\newcommand{\vars}[1]{\operatorname{Vars}(#1)}

\newcommand{\setsize}[1]{\left\lvert #1 \right\rvert}

\newcommand{\cl}[1]{\operatorname{cl}(#1)}

\newcommand{\clnull}[3]{\operatorname{cl}_{#2}^{#3}(#1)}

\newcommand{\clnullcons}[2]{\clnull{#1}{\nullcons}{#2}}

\newcommand{\polyp}{p}
\newcommand{\polyq}{q}
\newcommand{\polyr}{r}

\newcommand{\vertv}{v}

\newcommand{\nbhdub}[2]{N_{#1}(#2)}
\newcommand{\nbhstd}[2]{N_{#1}(#2)}

\newcommand{\spsize}{s}

\newcommand{\reststd}{\rho}
\newcommand{\restrict}[2]{{{#1}%
\upharpoonright_{#2}}}

\newcommand{\setT}{T}

\newcommand{\setW}{W}
\newcommand{\setU}{U}

\newcommand{\rels}{\mathbf{S}}

\newcommand{\relt}{\mathbf{T}}
\newcommand{\rela}{\mathbf{A}}
\newcommand{\relb}{\mathbf{B}}

\renewcommand{\hom}{\operatorname{Hom}}
\newcommand{\cspnotation}{\mathcal{P}}
\newcommand{\csp}[1]{\operatorname{CSP}(#1)}
\newcommand{\pcsp}[2]{\operatorname{PCSP}(#1, #2)}
\newcommand{\cspdist}[3]{\mathbb{A}_{#1, #2, #3}}
\newcommand{\lex}[2]{\operatorname{Lex}(#1;#2)}

\usepackage[a4paper,
   top=1in,
    bottom=1in,
    left=1in,
    right=1in
]{geometry}

\title{Lower Bounds for CSP Hierarchies Through Ideal Reduction\thanks{This is the publically available version of a paper with the same title to appear in \emph{SODA~'26}.}}

\author{Jonas Conneryd\thanks{Lund University. \texttt{\href{mailto:jonas.conneryd@cs.lth.se}{jonas.conneryd@cs.lth.se}}} \qquad Yassine Ghananne\thanks{University of Copenhagen and Lund University. \texttt{\href{mailto:yagh@di.ku.dk}{yagh@di.ku.dk}}} \qquad Shuo Pang\thanks{University of Bristol. \texttt{\href{mailto:s.pang@bristol.ac.uk}{s.pang@bristol.ac.uk}}}
}
\date{\today}

\begin{document}
\maketitle

\thispagestyle{empty}

\begin{abstract}

\noindent
We present a generic way to obtain level lower bounds for (promise) CSP hierarchies from degree lower bounds for algebraic proof systems. More specifically, we show that pseudo-reduction operators in the sense of Alekhnovich and Razborov [Proc. Steklov Inst. Math. 2003] can be used to fool the cohomological $\kcons$\nobreakdash-consistency algorithm. As applications, we prove optimal
level lower bounds for $\colcons$~vs.~$\ell$-coloring for all $\ell \geq \colcons \geq 3$, and give a simplified proof of the lower bounds for lax and null-constraining CSPs of Chan and Ng [STOC 2025].

\end{abstract}

\section{Introduction}

The class of \emph{constraint satisfaction problems (CSPs)} is a cornerstone of research in logic, algorithms, and complexity theory, encompassing a wide range of natural computational tasks such as Boolean satisfiability, graph $\kcons$-colorability, and solving systems of linear equations. Despite their generality, CSPs possess enough structure to enable a detailed understanding of which variants are tractable, which are not, and what fundamentally distinguishes the two. %

Every constraint satisfaction problem can be formulated as the task of determining, given two relational structures $\rela$ and $\relt$, whether there is a homomorphism $\rela\to \relt$ \cite{Jeavons98Algebraic,FV98Computational}. A widely studied variant is when the structure $\relt$ is kept fixed; $\relt$ is then called the \emph{template} of the CSP. %
A significant amount of research on CSPs has been motivated by the \emph{dichotomy conjecture} of Feder and Vardi~\cite{FV98Computational}. Broadly generalizing classical results of Schaefer~\cite{Schaefer78Dichotomy} and Hell and Ne\v{s}et\v{r}il~\cite{HN90HColoring}, the conjecture states that every fixed-template CSP is either in $\P$ or $\NP$\nobreakdash-complete. 

The dichotomy conjecture was resolved independently by Bulatov~\cite{Bul17Dichotomy} and Zhuk~\cite{Zhuk20Dichotomy}, as the culmination of two decades of efforts into the so-called \emph{algebraic approach} to CSPs; for an overview of this area, see \cite{BKW17Polymorphisms}. %
Informally, the results of Bulatov and Zhuk show that a CSP is tractable if and only if its template admits a certain type of non-trivial symmetry, or \emph{polymorphism}. Any CSP lacking such symmetries was previously known to be $\NP$-hard~\cite{MM08Existence}, and Bulatov and Zhuk exploit such a symmetry to construct a polynomial-time algorithm for the CSP. 

Following the resolution of the dichotomy conjecture, a new research direction has emerged around a broad generalization of CSPs, known as \emph{promise CSPs} or \emph{PCSPs}. An overview of this area can be found in~\cite{KO22InvitationPCSP}. Roughly speaking, the task in a PCSP is to decide whether an instance is satisfiable in a strong sense, or not even satisfiable in a weaker sense. More formally, given a pair of relational structures $(\rels, \relt)$ such that $\rels \to \relt$ and an instance structure $\rela$, the \emph{decision version} of the PCSP with template $(\rels, \relt)$ is the task of determining whether $\rela \to \rels$ or $\rela \not\to \relt$, promised one of the two holds. In the \emph{search version}, the task is to find a homomorphism $\rela \to \relt$ promised that $\rela \to \rels$.

The formal study of PCSPs was initiated by Austrin, Guruswami, and H{\aa}stad~\cite{AGH17SATNPHard}, but special cases have been investigated for much longer. For instance, perhaps the most well-known PCSP is \emph{approximate graph coloring}~\cite{GJ76Approximate}, also called \emph{$\colcons$~vs.~$\ell$-coloring}, where $\ell \geq \colcons \geq 3$ and the task is to decide whether an input graph $G$ is $\colcons$-colorable or not even $\ell$-colorable, promised one of the two holds. %
There are large gaps in our understanding of even this very natural PCSP: approximate graph coloring is widely believed to be $\NP$-hard for all constants $\ell \geq \colcons \geq 3$, %
but $\NP$-hardness is known only for $\colcons$~vs.~$(2\colcons-1)$-coloring for $\colcons \geq 3$~\cite{BBKO21AlgebraicCSP}, and for $\colcons$~vs.~\smash{$\big(\binom{\colcons}{\floor{\colcons/2}}-1\big)$}-coloring when $\colcons \geq 4$~\cite{KOWZ23Topology}. 
By contrast, the best known polynomial-time algorithm for coloring $3$-colorable graphs guarantees only an $\smash{O(n^{0.19747})}$-coloring for an $n$-vertex graph~\cite{KTY24BetterColoring}. %

In general, the existing algorithms for solving (P)CSPs, including those of Bulatov and Zhuk, can be said to consist of two main building blocks: \emph{local consistency}, and solving systems of \emph{linear equations}. Local consistency algorithms check whether small subproblems of the original CSP are satisfiable, and then enforce that the local solutions are consistent with each other. These checks take various forms: for instance, the well-known \emph{$\kcons$\nobreakdash-consistency} algorithm %
accepts an instance if and only if a consistent collection of partial homomorphisms with domain size at most $\kcons$ exists. The more stringent \emph{Sherali\nobreakdash--Adams hierarchy} \cite{SA90Hierarchy} 
of linear programming relaxations asks for a probability distribution on partial solutions that is consistent with the marginals of a distribution over satisfying assignments. The \emph{sum-of-squares} or \emph{Lasserre} hierarchy~\cite{Parrilo00Thesis, Lasserre01Explicit} of semidefinite programming relaxations further demands that the probabilities assigned to the local solutions satisfy certain orthogonality conditions when viewed as vectors. 

Algorithms that use linear equations to solve (P)CSPs are in turn based on the observation that any fixed-template CSP can be formulated as a system of linear equations over Boolean variables. %
These algorithms then relax the variable domain to one where, e.g., Gaussian elimination can be used to efficiently search for a \emph{global} solution over the relaxed domain. A prominent example is the family of \emph{affine relaxations}, which solve (variations of) the linear system formulation of CSPs over the integers. %
This can be done efficiently using, e.g., the classical Smith normal form, see \cite{Lazebnik96Systems}. Examples include the \emph{affine integer programming (AIP)}~\cite{BG19AlgorithmicBlend,BBKO21AlgebraicCSP} and \emph{linear Diophantine equation}~\cite{BG17Group} relaxations. 

It is well known that neither local consistency checking nor the linear equation paradigm alone solves all tractable CSPs, and both Bulatov's and Zhuk's %
algorithms use elaborate combinations of the two. In recent years, quite some work has gone into designing and analyzing simpler CSP algorithms than those of Bulatov and Zhuk. Besides providing a more fine-grained understanding of CSP families, the ultimate, if distant, goal of this line of work is to find a \emph{uniform} CSP algorithm---that is, one where the template $\relt$ is part of the input---which runs in polynomial time whenever $\csp{\relt}$ is tractable. %
In addition to its intrinsic importance, such an algorithm might also help establish the border of tractability for more general problems, such as PCSPs. The most well-studied algorithms in this line of work are so-called \emph{hierarchies} of relaxations, each defined with respect to a parameter~$\kcons$ called its \emph{level}. Larger values of $\kcons$ yield increasingly tighter relaxations, while the $\kcons$-th level of a hierarchy runs in time $m^{O(\kcons)}$ on an instance of size~$m$ (with some caveats in the case of semidefinite programming relaxations, see~\cite{ODonnell17SOSNotAutomatizable}). %

\subsection{CSP Hierarchies}
\label{sec:csp-hierarchies-intro}

Several authors have introduced and analyzed (P)CSP hierarchies that combine convex and affine relaxations. Brakensiek, Guruswami, Wrochna, and \v{Z}ivn\'{y}~\cite{BGWZ20PowerCombined}, generalizing \cite{BG19AlgorithmicBlend}, combined the first level of the SA hierarchy with AIP and showed that the resulting algorithm, which they call BLP+AIP, solves all \emph{Boolean} tractable CSPs. %
A generalization of BLP+AIP to higher levels of the SA hierarchy, called the \emph{$\mathit{BA}^\kcons$ hierarchy}, was suggested in~\cite{BGWZ20PowerCombined} and analyzed by Ciardo and \v{Z}ivn\'{y}~\cite{CZ25CrystalHollow}; see also~\cite{BG17Group}. 
The \emph{CLAP} and \emph{C(BLP+AIP)} algorithms due to Ciardo and \v{Z}ivn\'{y}~\cite{CZ23CLAP} refine BLP+AIP by iteratively pruning variable domains and solving the pruned relaxation.  %
Both CLAP and C(BLP+AIP) can be strengthened to hierarchies of relaxations, as noted in~\cite{CZ23CLAP} and~\cite{CN25Hierarchies}, respectively. %

Ciardo and \v{Z}ivn\'{y}~\cite{CZ25Semidefinite} and Chan, Ng, and Peng~\cite{CNP24Hierarchies} independently defined and analyzed a combination of the sum-of-squares and AIP hierarchies. We refer to this algorithm as the \emph{$\mathit{SDA}^\kcons$ hierarchy} in this paper, following~\cite{CZ25Semidefinite}. A related but distinct algorithm, which we do not discuss further here, was recently introduced by Bhangale, Khot, and Minzer~\cite{BKM25ApproximabilityV} in the context of hardness of approximation.

Another line of work blends the $\kcons$\nobreakdash-consistency algorithm with affine relaxations. Dalmau and Opr\v{s}al~\cite{DO19LocalConsistency} defined the \emph{$\Z$-affine $\kcons$\nobreakdash-consistency} algorithm, which first finds a consistent set $\kappa$ of partial homomorphisms through the $\kcons$\nobreakdash-consistency algorithm and then searches for a solution supported on $\kappa$ to an affine relaxation. %
\'{O}~Conghaile~\cite{OConghaile22Cohomology} introduced \emph{cohomological $\kcons$\nobreakdash-consistency}, which is a strengthened version of $\Z$-affine $\kcons$\nobreakdash-consistency with additional requirements on the solutions to the affine relaxation; we describe this algorithm in more detail in \cref{sec:discussion-of-pf-techniques}.%

Recently, Lichter and Pago~\cite{LP24Limitations} exhibited a tractable CSP that the $\mathrm{BA}^\kcons$ hierarchy, $\Z$\nobreakdash-affine $\kcons$\nobreakdash-consistency algorithm, and (even parameterized versions of) CLAP all fail to solve. Hence, $\mathrm{SDA}$, C(BLP+AIP), and the cohomological $\kcons$\nobreakdash-consistency algorithm are the remaining candidates for a uniform CSP algorithm among those mentioned in this section. %

\subsection{Lower Bounds for CSP Hierarchies}

An integral part of the design and analysis of (P)CSP algorithms is to understand their limitations; this is especially relevant for PCSPs, where the border of tractability is not known. For hierarchies, such results take the form of lower bounds on the level required to solve a given (P)CSP. Besides disqualifying candidates for uniform CSP algorithms through showing they fail on tractable CSPs, level lower bounds may help identify recurring features of hard instances for different families of hierarchies, which is instructive both in algorithm design and for proving stronger hardness results.

Level lower bounds for $\colcons$~vs.~$\ell$-coloring are known for all hierarchies mentioned in \cref{sec:csp-hierarchies-intro}, for varying parameter regimes. The $\kcons$\nobreakdash-consistency algorithm does not solve $\colcons$~vs.~$\ell$\nobreakdash-coloring on $n$-vertex graphs for any $\ell \geq \colcons \geq 3$ and $\kcons \leq \Omega(n)$~\cite{AD22Width,CZ24Periodic, CN25Hierarchies}; this lower bound is asymptotically optimal. In a sequence of works, Ciardo and \v{Z}ivn\'{y}~\cite{CZ25CrystalHollow, CZ25Semidefinite} showed that no constant level of the AIP, BA, or SDA hierarchy solves $\colcons$~vs.~$\ell$-coloring for any $\ell \geq \colcons \geq 3$. The results of \cite{CZ25Semidefinite} apply also to the more general \emph{approximate graph homomorphism} problem. Chan, Ng, and Peng~\cite[Appendix~C]{CNP24Hierarchies} showed that the $\kcons$-th level of the SDA hierarchy does not solve $\colcons$~vs.$\smash{\binom{\colcons}{\floor{\colcons/2}}}$-coloring\footnote{Their result is stated for $\colcons$~vs.~\smash{$\binom{\colcons}{\floor{\colcons/2}}$}-coloring but seemingly uses \cite{KOWZ23Topology}, which only applies for ${\ell = \binom{\colcons}{\floor{\colcons/2}} - 1}$.} for $\colcons \geq 4$ and $\kcons \leq \Omega(n)$, and Chan and Ng~\cite{CN25Hierarchies} proved the same result for the cohomological $\kcons$\nobreakdash-consistency and C(BLP+AIP) hierarchies. The results of \cite{CNP24Hierarchies,CN25Hierarchies} are based on hardness reductions, and therefore seem unlikely to be improved beyond the strongest parameters for which $\NP$-hardness is known, established in~\cite{KOWZ23Topology}. %

In addition to their lower bounds for a tractable CSP mentioned above, Lichter and Pago~\cite{LP24Limitations} show that cohomological $\kcons$\nobreakdash-consistency correctly solves the CSP on which the other algorithms fail, for $\kcons \geq 4$. They also prove superconstant lower bounds for the cohomological $\kcons$\nobreakdash-consistency algorithm for solving a certain $\NP$\nobreakdash-complete CSP. Chan, Ng, and Peng~\cite{CNP24Hierarchies} and Chan and Ng~\cite{CN25Hierarchies} together establish optimal level lower bounds for all hierarchies in \cref{sec:csp-hierarchies-intro}. Their lower bounds are for random instances of so-called \emph{$\tau$-wise neutral} CSPs for SDA and weaker hierarchies, \emph{null-constraining CSPs} for the $\kcons$\nobreakdash-consistency algorithm, \emph{lax and pairwise uniform} CSPs for the C(BLP+AIP) hierarchy, and for \emph{lax and null-constraining} CSPs for cohomological $\kcons$\nobreakdash-consistency. %
To the best of our knowledge, no known non-trivial tractable CSP has the properties required for the lower bounds in \cite{CNP24Hierarchies,CN25Hierarchies}, except for those for the $\kcons$\nobreakdash-consistency algorithm.  %

\subsection{CSP Hierarchies and Proof Complexity}
\label{sec:csp-hierarchies-proof-complexity}

All hierarchies for $\csp{\relt}$ mentioned in \cref{sec:csp-hierarchies-intro} always accept instances $\rela$ such that $\rela \to \relt$. Therefore, to prove a lower bound on the level required to solve $\csp{\relt}$, we should exhibit a \emph{fooling instance}: a structure $\rela$ such that $\rela \not \to \relt$, but which is nevertheless accepted by the $\kcons$-th hierarchy level. For a PCSP with template $(\rels, \relt)$, a fooling instance is a structure $\rela$ such that $\rela\not\to\relt$, but such that the $\kcons$-th level of the hierarchy for $\csp{\rels}$ accepts $\rela$. 

This setting closely resembles that in \emph{proof complexity}, where the goal is to prove lower bounds on the complexity of certificates of unsatisfiability of a propositional formula; in this case, the formula would encode that $\rela\to\relt$. The proof complexity of CSPs has a long and extensive history, arguably going back to the work of Tseitin~\cite{Tseitin68Complexity}. However, it seems fair to say that the area started in earnest with the seminal paper of Chv\'{a}tal and Szemer\'{e}di~\cite{CS88ManyHard}, who showed that, with high probability, random $\kcons$-SAT requires exponential size to refute in the \emph{resolution} proof system. Their results were formative in showing that \emph{random} CSP instances are often hard, and %
by now, lower bounds for multiple families of random CSPs are known for resolution~\cite{Mitchell02Resolution, Mitchell02Random, BCCM05RandomGraph,BIS07IndependentSets, MS07RandomConstraint,CM13Dichotomy}, as well as for other proof systems. As two examples relevant to this work, we mention \emph{polynomial calculus}~\cite{BI10Random,AR03LowerBounds, MN24Generalized,CdRNPR25GraphColouring} and \emph{sum-of-squares}~\cite{Grigoriev01LinearLowerBound,Schoenebeck08LinearLevel,  BCM15SoSPairwise, BHKKMP19NearlyTight, KMOW2017Sumofsquares, JKM19Lovasz, KM21StressFree, JPRTX21Sparse, Pang21Exact, KPX24UltraSparse, PX25SoSColoring}.%

Multiple results and techniques from proof complexity have informed the study of lower bounds for CSP algorithms and vice versa; some examples include \cite{BG15LimitationsAlgebraic, BG17Group,LN17GraphColouring, AO19ProofCplx, CNP24Hierarchies, CN25Hierarchies}. %
The main result of this work provides yet another example of a fruitful connection between the two areas, by showing that a technique used for proving polynomial calculus degree lower bounds can be used to establish level lower bounds for the cohomological $\kcons$\nobreakdash-consistency algorithm (and hence also for weaker hierarchies).

Our results are perhaps most closely related to \cite{BG15LimitationsAlgebraic, BG17Group, AO19ProofCplx}, which relate polynomial calculus %
 degree lower bounds to hardness results for solving CSPs in various established algorithmic paradigms. However, in contrast to the algorithms in these works, lower bounds for cohomological $\kcons$-consistency do not follow in a black-box fashion from polynomial calculus degree lower bounds. For instance, the tractable CSP in the work of Lichter and Pago~\cite{LP24Limitations} is solvable by constantly many levels of cohomological $\kcons$-consistency, but requires linear polynomial calculus degree over every field.\footnote{This CSP is a certain encoding of a disjunction of two so-called \emph{Tseitin contradictions} on expander graphs over two different fields~\cite[Sections~3\nobreakdash-5]{LP24Limitations}. Tseitin contradictions over $\F_q$ on expander graphs require linear polynomial calculus degree over fields of characteristic different from $q$~\cite{BGIP01LinearGaps} and can be obtained from such CSP instances through variable restrictions. Hence, this CSP requires linear degree over every field.}

\subsection{Our Contributions}
\label{sec:contributions}

We show that the \emph{Alekhnovich--Razborov method}~\cite{AR03LowerBounds} for polynomial calculus degree lower bounds can be strengthened to prove level lower bounds for certain CSP hierarchies. 
Since this method has been successfully applied to many problems in proof complexity, e.g., \cite{AR03LowerBounds, GL10Automatizability, GL10Optimality, MN24Generalized,CdRNPR25GraphColouring}, we are hopeful that it will find use also in this new context.

As a first application, we prove an asymptotically optimal lower bound on the level required for the cohomological $\kcons$\nobreakdash-consistency algorithm to solve approximate graph coloring, building on the corresponding polynomial calculus degree lower bound in~\cite{CdRNPR25GraphColouring}. 
\begin{theorem}[Informal]
    \label{th:main-thm-informal}
    With high probability as $n\to \infty$, a random $d$-regular $n$-vertex graph has chromatic number at least $\graphdeg/4\log{\graphdeg}$ but is accepted by the cohomological $\kcons$\nobreakdash-consistency algorithm for $3$-colorability for all $\kcons \leq n\cdot \graphdeg^{-O(\graphdeg)}$. Consequently, the cohomological $\kcons$\nobreakdash-consistency algorithm does not solve $3$ vs. $\graphdeg/4\log{\graphdeg}$-coloring on $n$-vertex graphs for any $\kcons \leq  n\cdot \graphdeg^{-O(\graphdeg)}$. 
\end{theorem}
For a full statement of \cref{th:main-thm-informal}, see \cref{sec:approximate-graph-coloring}. Setting $\graphdeg$ to be a constant in \cref{th:main-thm-informal}, it follows that cohomological $\kcons$\nobreakdash-consistency does not solve $3$ vs. $\colcons$-coloring on $n$-vertex graphs for any constant $\colcons \geq 3$ and any $\kcons \leq \Omega(n)$; our lower bound on $\kcons$ is then optimal up to constant factors. On the other extreme, we obtain a superconstant lower bound on $\kcons$ for all $\graphdeg=o((\log{n})/\log{\log{n}})$. Our results improve those of Chan and Ng~\cite{CN25Hierarchies}, who show through reductions that cohomological $\kcons$\nobreakdash-consistency does not solve $c$~vs.~\smash{$\binom{c}{\floor{c/2}}$}-coloring for any $c\geq 4$ and $\kcons \leq \Omega(n)$; recall that for these parameters, approximate graph coloring is known to be $\NP$-hard~\cite{KOWZ23Topology}. %

As previously noted, the only remaining candidates for uniform CSP algorithms among those in \cref{sec:csp-hierarchies-intro} are cohomological $\kcons$\nobreakdash-consistency, SDA, and C(BLP+AIP). This fact makes \cref{th:main-thm-informal} noteworthy in light of the results of Banks, Kleinberg, and Moore~\cite{JKM19Lovasz}, who show that with high probability, semidefinite programming---and hence also algorithms based on it, such as the sum-of-squares and SDA hierarchies---correctly identifies that a random $d$\nobreakdash-regular graph has chromatic number at least $\sqrt{\graphdeg}/2$. Therefore, with high probability, the $\mathrm{SDA}$ hierarchy solves $\sqrt{\graphdeg}/2$~vs.~$\graphdeg/4\log{\graphdeg}$-coloring on a random $\graphdeg$-regular graph. Thus, \cref{th:main-thm-informal} yields a separation %
of SDA from cohomological $\kcons$\nobreakdash-consistency, albeit for a PCSP that is widely believed to be $\NP$-hard.

As another application, we give a different proof of the lower bounds for lax and null-constraining CSPs (for definitions, see~\cref{sec:lax-and-null-constraining}) in \cite[Theorem~1.3]{CN25Hierarchies}.%
\begin{theorem}[Informal, cf. \cite{CN25Hierarchies}, Theorem~1.3]
    \label{th:lax-null-cons-informal}
    If an $n$-variable $\uniform$-CSP $\cspnotation$ is non-trivial, lax, and null-constraining, then there exist $\Delta> 0$ and $\delta > 0$ such that, with probability ${\delta - o_n(1)}$, a randomly sampled instance of $\cspnotation$ with $\Delta n$ constraints is unsatisfiable but is accepted by the cohomological $\kcons$\nobreakdash-consistency algorithm for all $\kcons\leq \Omega(n)$. Consequently, the cohomological $\kcons$\nobreakdash-consistency algorithm does not solve $\cspnotation$ for any $\kcons\leq \Omega(n)$.
\end{theorem}

The full statement of \cref{th:lax-null-cons-informal} appears in \cref{sec:lax-and-null-constraining}. Since the proof of \cref{th:lax-null-cons-informal} goes through the connection to pseudo-reduction operators, we immediately obtain the same lower bound of $\Omega(n)$ for the polynomial calculus degree required to refute a random instance of a lax and null-constraining CSP; see \cref{cor:cohomological-lax-nullcons-lbs-pc}. %

\subsection{Technical Overview}
\label{sec:discussion-of-pf-techniques}

We now review the connection between polynomial calculus degree lower bounds and level lower bounds for cohomological $\kcons$\nobreakdash-consistency, formally established in Lemmas \ref{lem:accepted-by-z-aff-cons} and \ref{lem:cohomological-k-consistency-lower-bounds}. To this end, we first provide a brief overview of the necessary technical background.

\paragraph{Proof Complexity Background:}
The starting point of our connection is a technique for proving lower bounds for the \emph{polynomial calculus} proof system~\cite{CEI96Groebner,ABRW02SpaceComplexity}. 
Given a field~$\F$ and a set $\unsatp$ of polynomials in $\F[x_1, \ldots, x_n]$, polynomial calculus proves that the polynomials in $\unsatp$ have no common root through deriving new polynomials in the ideal $\ideal{\unsatp}$ generated by $\unsatp$ until reaching the constant polynomial~$1$.

The most common complexity measures of a polynomial calculus proof are \emph{size}, defined as the number of monomials %
in it, and \emph{degree}, defined as the maximum degree among all polynomials appearing in the proof. It is known that a strong enough lower bound on degree implies lower bounds on size in a black-box fashion~\cite{IPS99LowerBounds}; indeed, almost all known polynomial calculus size lower bounds are proved in this way. Barring some early results~\cite{ BGIP01LinearGaps, Grigoriev01LinearLowerBound, BI10Random, BCIP02Homogenization}, almost all polynomial calculus degree lower bounds are in turn proved by constructing a so-called \emph{pseudo-reduction operator} or \emph{$R$-operator}~\cite{Razborov98LowerBound}, which is an $\F$-linear operator that sends all low-degree polynomials derived from $\unsatp$ to $0$ but sends $1$ to $1$; thus, there can be no low-degree polynomial calculus proof that $\unsatp$ is unsatisfiable. %

Pseudo-reduction operators are often based on the notion of \emph{reduction modulo ideals} from Gr\"{o}bner basis theory. Given a 
so-called \emph{admissible} %
total order $\prec$ on the monomials in $\F[x_1, \ldots, x_n]$ and an ideal $I \subseteq \F[x_1, \ldots, x_n]$, %
the \emph{reduction operator $\redop_I(\cdot)$} on $\F[x_1, \ldots, x_n]$ %
acts on a polynomial~$p$ by iteratively identifying a term $t$ in $p$ and a polynomial $q \in I$ whose largest term with respect to $\prec$ is $t$, subtracting $q$ from $p$, and repeating until no such $q$ exists. It then returns the resulting polynomial.
 As an analogy (and indeed a special case), one can think of a system of linear equations being reduced to row-echelon form. 
 
 If the polynomials in a set $\unsatp$ have a common root, then $\pseudred_{\langle\unsatp\rangle}$ is a pseudo-reduction operator: all polynomials derived from $\unsatp$, of any degree, belong to $\ideal{\unsatp}$ and are thus reduced to $0$, and since the polynomials in $\unsatp$ have a common root, we have that $1 \not \in \ideal{\unsatp}$ and hence $\redop_{\ideal{\unsatp}}(1) = 1$. A pseudo-reduction operator $\pseudred$ should therefore behave as $\pseudred_{\langle\unsatp\rangle}$ would on low-degree monomials if the polynomials in $\unsatp$ had a common root. 
In line with this heuristic, Alekhnovich and Razborov \cite{AR03LowerBounds} pioneered an approach for constructing pseudo-reduction operators based on \emph{local} reductions: to each low-degree monomial $m$, associate a satisfiable subset $S(m)$ of $\unsatp$, define $\pseudred(m) = \redop_{\ideal{S(m)}}(m)$, and define $\pseudred$ on arbitrary polynomials by linear extension. 
We refer to operators constructed in this way as \emph{local} pseudo-reduction operators.%

\paragraph{Cohomological \textit{\kcons}-Consistency:}

Given an instance $\rela$ of $\csp{\relt}$, the cohomological $\kcons$-consistency algorithm maintains a collection $\mathcal{H} = \{\mathcal{H}(X) : X \in \smash{\binom{\ala}{\leq \kcons}}\}$ of partial homomorphisms $\rela[X] \to \relt$, which is iteratively refined in two steps, outlined below. When $\mathcal{H}$ stabilizes, the algorithm rejects $\rela$ if $\mathcal{H}(X) = \varnothing$ for some $X \in \smash{\binom{\ala}{\leq \kcons}}$ and accepts otherwise. 

A \emph{$\kcons$\nobreakdash-consistent} collection $\kappa$ of partial homomorphisms is one where (1) for each $\homomphi \in \kappa$ with domain size %
smaller than $\kcons$, 
there is some $\psi_S \in \kappa$ extending $\homomphi$ to any larger domain $S$ of size  $k$, and (2) all restrictions of all $\homomphi \in \kappa$ to smaller domains are also in $\kappa$. In the first step of each iteration, the algorithm finds the maximal $\kcons$\nobreakdash-consistent collection $\kappa$ in~$\mathcal{H}$ and assigns $\mathcal{H} = \kappa$. 

In the second step, the algorithm seeks a family of solutions over 
$\Z$ to a linear system formulation 
$L_\kcons(\rela, \relt)$ of $\csp{\relt}$. %
The system $L_k(\rela, \relt)$ has variables $x_\homomphi$ for all partial homomorphisms $\homomphi \colon \rela \to \relt$ with domain size at most $\kcons$. The solutions to $L_\kcons(\rela, \relt)$ must all be supported on $\kappa$, and for each $\homomphi \in \kappa$ there must be a solution that sets $x_\homomphi$ to $1$ and $x_\homompsi$ to $0$ for all $\homompsi \in \kappa$ with the same domain as $\homomphi$. If no such solution exists, then $\homomphi$ is removed from $\mathcal{H}$ in the next iteration.

\paragraph{A Link Between the Two:}
Homomorphisms $\rela \to \relt$ can be encoded as the common roots of a set $\polyhomo{\rela}{\relt}$ of polynomials over $\F$ in Boolean variables $\{x_{\verta,\verti}\mid \verta\in \ala, \verti\in \alt\}$. Any (partial) function $\homomphi : \ala \to \alt$ can then be encoded as a monomial $m_\homomphi = \prod_{(\verta, \verti) \in \homomphi} x_{\verta, \verti}$. 

Suppose that $\rela \not \to \relt$, that $\F$ has characteristic $0$, and that $\polyhomo{\rela}{\relt}$ admits a degree-$\degstd$ %
local pseudo-reduction operator as above, where for each monomial $m$ we associate a substructure $\rela(m)$ of $\rela$ and define $R(m) = R_{\ideal{\polyhomo{\rela(m)}{\relt}}}(m)$. We use $\pseudred$ to show that $\rela$ is then a fooling instance for the cohomological $\kcons$\nobreakdash-consistency algorithm for all $\kcons\leq \degstd$. More specifically, we show that:

\begin{enumerate}
    \item the functions $\homomphi: \ala \to \alt$ such that $\pseudred(m_\homomphi) \neq 0$ form a $\degstd$-consistent collection of partial homomorphisms;
    \label{item:consistent}
    \item the assignment $x_\homomphi \mapsto \pseudred(m_\homomphi)$ is a solution to $L_k(\rela, \relt)$, except over the polynomial ring $\F[\mathbf{x}_{\ala, \alt}]$ instead of~$\mathbb Z$; %
    \label{item:solutions}

    \item for each monomial $m$ of degree at most $\degstd$, the polynomial $\pseudred(m)$ has integer coefficients.
    \label{item:integer-coefficients}
\end{enumerate}
Consequently, composing the map $x_\homomphi \mapsto \pseudred(m_\homomphi)$ with 
any evaluation on a point in $\Z^{\setsize{\ala} \cdot \setsize{\alt}}$, %
or more generally with any linear map $\Z[\mathbf{x}_{\ala, \alt}] \to \Z$ such that $h(1) = 1$, %
yields a solution to $L_k(\rela, \relt)$ over $\Z$. As we show in \cref{lem:cohomological-k-consistency-lower-bounds}, the specific properties of %
local pseudo-reduction operators lets us judiciously choose these maps to produce the family of solutions to $L_k(\rela, \relt)$ sought by the cohomological $\kcons$\nobreakdash-consistency algorithm.

Items \ref{item:consistent} and \ref{item:solutions} are established in \cref{lem:accepted-by-z-aff-cons} and follow rather naturally from the definition of $\polyhomo{\rela}{\relt}$ and general properties of pseudo-reduction operators. 
The integrality of $\pseudred(m)$ in \cref{item:integer-coefficients} is proved in \cref{lem:integer-reduction}, and the key tool we use is the \emph{lex game} \cite{FRR06LexGame}. The lex game over a polynomial ring $\F[x_1, \ldots, x_n]$ with a lexicographic monomial order $\prec$ has two players named Lea and Stan, and is defined with respect to a set $V\subseteq \F^n$ and monomial~$m$. 

It is shown in~\cite{FRR06LexGame} that Lea has a winning strategy in the lex game if and only if the ideal $I(V)$ of polynomials that vanish on $V$ contains a polynomial with largest term $m$. %
The key insight for showing %
\cref{item:integer-coefficients} is that this winning strategy can itself be encoded as a polynomial in $I(V)$, which has largest term $m$ and has integer coefficients when $V \subseteq \{0, 1\}^n$. 
Therefore, each reduction step performed by $\redop_{I(V)}$ can be done using a polynomial with integer coefficients, and hence the output is integral. %
Since the ideals $\ideal{\polyhomo{\rela(m)}{\relt}}$ are over Boolean variables, our result applies to them, and so \cref{item:integer-coefficients} follows.

Given the generality of our connection and the flexibility it offers in constructing solutions to the cohomological $\kcons$\nobreakdash-consistency algorithm%
, we expect %
it to be useful for proving lower bounds for other (P)CSPs, and potentially also for other hierarchies. %

\paragraph{Applications:} To prove \cref{th:main-thm-informal}, we apply our connection to the polynomial calculus degree lower bound in \cite{CdRNPR25GraphColouring}, which uses %
local pseudo-reduction operators. 
For~\cref{th:lax-null-cons-informal}, we construct a %
local pseudo-reduction operator for $\polyhomo{\rela}{\relt}$ by using the \emph{BW closure} from \cite{CN25Hierarchies} to define the %
subset $S(m)$ of constraints for each low-degree monomial $m$. As detailed in \cref{lem:ipr-alekhnovich-razborov}, %
this amounts to showing that the so-called \emph{satisfiability} and \emph{reducibility} conditions are satisfied for the BW closure. Roughly speaking, the satisfiability condition asks that %
$S(m)$ is satisfiable, %
and the reducibility condition states that $\pseudred_{\ideal{S(m)}} (m) = \pseudred_{\ideal{S(m) \union T}} (m)$ for any small set $T \subseteq \polyhomo{\rela}{\relt}$. We establish these conditions in Lemmas~\ref{lem:extension-lemma} and \ref{lem:reducibility-nullcons} respectively, using many ideas of \cite{CN25Hierarchies}.

\subsection{Paper Organization}
\label{sec:paper-organization}
The rest of this paper is organized as follows. In \cref{sec:prelims} we present the necessary preliminaries. Sections~\ref{sec:csp-hiearchies-ipr-operators} and \ref{sec:ar-method-cohomological-consistency} establish our connection between lower bounds for polynomial calculus degree and CSP hierarchy level. In \cref{sec:ipr-operators-lex-game} we introduce the lex game and prove our key results regarding integrality of reductions. We prove our lower bounds for approximate graph coloring in \cref{sec:approximate-graph-coloring}, and for lax and null-constraining CSPs in \cref{sec:lax-and-null-constraining}. Finally, we present some concluding remarks and open problems in \cref{sec:concluding-remarks}.

\section{Preliminaries}
\label{sec:prelims}
This section reviews the necessary preliminaries from algebra, constraint satisfaction, and proof complexity. For a positive integer $n$, we write $[n]$ to denote the set $\{1, 2, \ldots, n\}$. Tuples are denoted by boldface lowercase letters. %
We say a set $S$ is \emph{$a$-small} if $\setsize{S}\leq a$. We denote the family of subsets of $S$ by $2^S$, the family of subsets of $S$ of size $\kcons$ by $\binom{S}{\kcons}$, and the family of subsets of size at most $\kcons$ by $\binom{S}{\leq \kcons}$.

We use standard asymptotic notation: for functions $f, g\colon \N \to \mathbb{R}$, we write $f = O(g)$ if there exist nonnegative constants $c$ and $n_0$ such that $f(n) \leq c\cdot g(n)$ for all $n \geq n_0$, and write $f = \Omega(g)$ if $g = O(f)$. If $f=O(g)$ and $g = O(f)$, we write $f = \Theta(g)$. 

\subsection{Relational Structures and Constraint Satisfaction Problems}

A \emph{signature} $\sigma$ is a finite set of relation symbols $\relsymbol_1, \ldots, \relsymbol_\ell$, each with respective arity $\ar(\relsymbol_i) \in \N$. %
A \emph{$\sigma$-structure} $\rela$ consists of a domain $\ala$ and, for each $\relsymbol \in \sigma$, a relation $\relsymbol^{\rela} \subseteq \ala^{\ar(\relsymbol)}$. %
A \emph{relational structure} is a $\sigma$-structure, for some signature $\sigma$. For a $\sigma$-structure $\rela$ and a set $X\subseteq A$, the \emph{substructure of $\rela$ on $X$}, denoted $\rela[X]$, is the $\sigma$-structure with domain $X$ where, for $\Gamma\in \sigma$, the relation $\relsymbol^{\rela[X]}$ consists of the tuples in $\relsymbol^{\rela}$ that mention only variables in $X$. A relational structure $\rela$ is \emph{finite} if $\setsize{\ala}$ is finite. All relational structures in this paper are finite unless otherwise stated.

For two $\sigma$-structures $\rela$ and $\relb$, a \emph{homomorphism} from $\rela$ to $\relb$ is a map $\homomphi \colon A\to B$ such that for each symbol $\relsymbol \in \sigma$ and each tuple $\mathbf{\verta} = (\verta_1, \ldots, \verta_{\ar(\relsymbol)})$, it holds that if $\mathbf{\verta} \in \relsymbol^{\rela}$ then $\homomphi(\mathbf{\verta}) = (\homomphi(\verta_1), \ldots, \homomphi(\verta_{\ar(\relsymbol)})) \in \relsymbol^{\relb}$. If there exists a homomorphism from $\rela$ to $\relb$, we write $\rela \to \relb$; if not, we write $\rela \not \to \relb$. The restriction of a homomorphism $\homomphi \colon \rela \to \relt$ to a set $X\subseteq \ala$ is denoted $\restrict{\homomphi}{X}$. Given sets $Y \subseteq X \subseteq \ala$ and a homomorphism $\homomphi \colon \rela[Y] \to \relt$, if there exists a homomorphism $\homompsi \colon \rela[X] \to \relt$ such that $\restrict{\homompsi}{Y} = \homomphi$, we say that $\homompsi$ \emph{extends} $\homomphi$. We denote the set of homomorphisms from $\rela$ to $\relb$ by $\hom(\rela, \relb)$. 

For a $\sigma$-structure $\relt$, an \emph{instance of the constraint satisfaction problem with template $\relt$} is a $\sigma$-structure $\rela$ such that, for all $\relsymbol \in \sigma$, there are no tuples in $\relsymbol^\rela$ with repeated elements.\footnote{This technical condition corresponds to having no repeated variables in the scope of any constraint in the instance.} We define $\csp{\relt}$ as the computational problem where the input is an instance $\rela$ and the task is to determine whether $\rela\to\relt$.

For a pair $(\rels, \relt)$ of fixed $\sigma$-structures such that $\rels \to \relt$, the \emph{promise constraint satisfaction problem (PCSP) with template $(\rels, \relt)$}, denoted by $\pcsp{\rels}{\relt}$, is the computational task of distinguishing whether $\rela \to \rels$ or $\rela \not\to\relt$ for a given instance $\rela$, promised that one of the two holds.\footnote{This is the \emph{decision} variant of PCSP. In the \emph{search} variant, the input is an instance $\rela$ such that $\rela \to \rels$, and the task is to find a homomorphism from $\rela$ to $\relt$. For CSPs, it is a standard fact that the search version reduces to the decision version, but no generic search-to-decision reduction is known for PCSPs. Since an efficient algorithm for the search version can be used to also solve the decision version, a lower bound for the latter also holds for the former.}  Note that $\pcsp{\relt}{\relt} = \csp{\relt}$.

\subsection{Graphs and Hypergraphs}
A \emph{hypergraph} consists of a finite vertex set $V$ together with a collection $\hypedgeset$ of subsets of $V$. The elements of $\hypedgeset$ are called \emph{hyperedges}. A hyperedge set $\hypedgeset$ is \emph{$\uniform$-uniform} if all hyperedges in $\hypedgeset$ contain exactly $\uniform$ vertices, and a hypergraph $H=(V, \hypedgeset)$ is $\uniform$-uniform if its hyperedge set is. For a hyperedge set $\hypedgeset$, we denote the set $\bigcup_{e \in \hypedgeset} e$ of vertices in $\hypedgeset$ by $V(\hypedgeset)$. For a vertex $v\in V$, the \emph{degree} $\deg_{\hypedgeset}(v)$ of $v$ in $\hypedgeset$ is the number of hyperedges in $\hypedgeset$ containing $v$. The subscript $\hypedgeset$ is usually omitted if the hyperedge set is clear from context.  A hypergraph is \emph{$\graphdeg$-regular} if all its vertices have degree exactly $\graphdeg$. For a set $U\subseteq V$ and a hyperedge set $\hypedgeset$, the \emph{neighborhood} of $U$ in $\hypedgeset$ is $\nbhdub{\hypedgeset}{U} = \{v \in V(\hypedgeset)\setminus U : \exists u \in U, \exists e \in \hypedgeset \mid \{u, v\} \subseteq e\}$. For a hypergraph $H = (V, \hypedgeset)$ and $U \subseteq V$, the subhypergraph of $H$ induced by $U$ is $H[U] = (U, \hypedgeset[U])$, where $\hypedgeset[U] = \{e \in \hypedgeset : e \subseteq U\}$. 

A \emph{walk of length $\ell$} in a hyperedge set $\hypedgeset$ is a sequence of vertices $v_0, \ldots, v_\ell$ and a sequence of hyperedges $e_0, \ldots, e_{\ell-1}$ such that $\{v_i, v_{i+1}\} \subseteq e_i$ for all $0\leq i < \ell$. A walk is a \emph{(Berge) path} if $v_0,\ldots, v_{\ell}$ are distinct and $e_0, \ldots, e_\ell$ are also distinct; the vertices $v_0,\ldots, v_\ell$ are then the \emph{connecting vertices} of the path, and $v_0$ and $v_\ell$ are its \emph{endpoints}. A walk is a \emph{(Berge) cycle} if (1) $\ell \geq 2$, (2) $v_0, \ldots, v_{\ell-1}$ are distinct, (3) $v_\ell = v_0$, and (4) $e_0, \ldots, e_{\ell-1}$ are distinct. 
The \emph{girth} of a hyperedge set $\hypedgeset$, denoted $\operatorname{girth}(\hypedgeset)$, is the minimum length of a cycle in $\hypedgeset$, or $\infty$ %
if there are no cycles in $\hypedgeset$. 
A hyperedge $e_i$ in a path is \emph{simple} if every vertex in $e_i \setminus \{v_i, v_i+1\}$ has degree $1$ in $\{e_j :0 \leq j < \ell\}$, and a connecting vertex $v_i$ in a path is simple if it has degree $2$ in $\{e_j :0 \leq j < \ell\}$. A path is simple if all its hyperedges and non-endpoint connecting vertices are simple. A cycle is simple if all its hyperedges and connecting vertices are. 
A path is \emph{pendant} in a hyperedge set $\hypedgeset$ if none of its non-endpoints belong to any hyperedge in $\hypedgeset \setminus \{e_0, \ldots, e_{\ell-1}\}$.

A \emph{proper $c$-coloring} of a hypergraph $H = (V, \hypedgeset)$ is a map $\chi \colon V \to [c]$ such that $\setsize{\chi(e)} \geq 2$ for all $e \in \hypedgeset$. If there exists a proper $c$-coloring of $G$, we say that $G$ is \emph{$c$-colorable}. The \emph{chromatic number} of a hypergraph $G$ is the smallest $c$ such that $G$ is $c$-colorable. %

A \emph{graph} $G= (V, E)$ is a $2$-uniform hypergraph; in particular, all hypergraph notions in this section specialize to their usual graph-theoretic meaning.

\subsection{Polynomials and Reduction Modulo Ideals}
\label{sec:reduction-modulo-ideals}
For a ring $\ring$, we denote the polynomial ring over $\ring$ in $n$ variables by $\ring[x_1, \ldots, x_n]$, and the set of variables appearing in a polynomial $\polyp$ by $\vars{p}$. A \emph{monomial} is a product of variables. For a vector $\alpha = (\alpha_1, \ldots, \alpha_n) \in \N^n$, we write $x^\alpha$ to denote the monomial $\prod_{i=1}^n x_i^{\alpha_i} \in \ring[x_1, \ldots, x_n]$. Given a monomial $m$, its \emph{multilinearization} $m_\mathrm{mult}$ is the product of the variables appearing in $m$. For a set of polynomials~$ \unsatp=\{\polyp_1, \ldots, \polyp_m\}$ in $\ring[x_1, \ldots, x_n]$, the set of polynomials of the form $\sum_{i \in [m]}\polyq_i \polyp_i$ for $\polyq_i \in \ring[x_1, \ldots, x_n]$ is called the \emph{ideal generated by $\unsatp$} and is denoted by~$\ideal{\unsatp}$. For a set $S\subseteq \ring^n$, we denote the ideal of polynomials vanishing on all of $S$ by $I(S)$. %
If $S$ is empty, we define $I(S)$ to be $\ring[x_1, \ldots, x_n]$. %
For a polynomial $\polyp$ and a (possibly partial) mapping $\rho\colon\vars{\polyp} \to \ring$, we denote the mapping $\ring[x_1, \ldots, x_n] \to \ring[x_1, \ldots, x_n]$ that evaluates $\pcpolyp$ on $\reststd$ by $\operatorname{eval}_{\reststd}$, and the polynomial obtained by evaluating $\polyp$ according to $\rho$ by $\restrict{\polyp}{\rho}$. %
  
  For a field $\F$, a total well-order $\prec$ on the monomials in $\F[x_1, \ldots, x_n]$ is \emph{admissible} if (1) for any monomial $m$ with at least one variable, it holds that $1 \prec m$, and (2) for any monomials $m_1, m_2$ and $m$ such that $m_1 \prec m_2$, it holds that $m \cdot m_1 \prec m \cdot m_2$. A prominent example which will be central to this work is the \emph{lexicographic} order, where $x_1 \succ \ldots \succ x_n$ and for any two monomials $x^{\alpha_1}, x^{\alpha_2} \in \F[x_1, \ldots, x_n]$, it holds that $x^{\alpha_1} \succ x^{\alpha_2}$ if and only if the leftmost entry of the vector difference $\alpha_1 - \alpha_2 \in \Z^n$ is positive. 
  
  Fixing an admissible order $\prec$ on the monomials in $\F[x_1, \ldots, x_n]$, the \emph{leading term} of a polynomial is its term with largest monomial with respect to $\prec$. A polynomial $\polyp$ is \emph{reducible modulo} an ideal~$I$ if there exists $\polyq\in I$ with the same leading term as $\polyp$; otherwise, $\polyp$ is \emph{irreducible} modulo~$I$. It is straightforward to show that for any polynomial $\polyp$ and any ideal $I$, there exists a unique representation $\polyp = \polyq + \polyr$ where $\polyq\in I$ and where $\polyr$ is a linear combination of irreducible monomials modulo~$I$. The polynomial $r$ is called the \emph{reduction} of $\polyp$ modulo $I$, and the operator $R_I$ on $\F[x_1, \ldots, x_n]$ that takes a polynomial $\polyp$ to its reduction modulo $I$ is called the \emph{reduction operator modulo $I$}. A straightforward argument shows that the reduction operator is $\F$-linear.

Finally, we recall the Boolean Nullstellensatz, omitting its standard proof (see, e.g., \cite[Appendix~A]{CdRNPR25GraphColouring}).

\begin{lemma}[Boolean Nullstellensatz]
  \label{lem:boolean-ns}
  For any finite set $\unsatp \subseteq \F[x_1, \ldots, x_n]$ that contains the set of Boolean axioms $\{x_1^2-x_1, \ldots, x_n^2-x_n\}$, and any polynomial $\pcpolyq \in \F[x_1, \ldots, x_n]$, %
  $\pcpolyq$ vanishes on all common roots of the polynomials in $\unsatp$ if and only if $\pcpolyq\in \ideal{\unsatp}$. In particular, $1 \in \ideal{\unsatp}$ if and only if the polynomials in $\unsatp$ have no common root.%
\end{lemma}
\subsection{Algebraic Proof Complexity}
\label{sec:proof-complexity}

\emph{Nullstellensatz}~\cite{BIKPP94LowerBounds} is a proof system
that uses algebraic reasoning to prove that a set $\unsatp = \{p_1, \ldots, p_m\}$ of
polynomials in variables $x_1, \ldots, x_n$ over a field $\F$ is unsatisfiable, that is, that the polynomials in $\unsatp$ have no common root in $\F$. We often refer to $\unsatp$ as the set of \emph{axioms},
and we say that a subset of axioms $\mathcal{Q}\subseteq\unsatp$
is \emph{satisfiable} if the polynomials in $\mathcal{Q}$ have a 
common root. A \emph{Nullstellensatz refutation} of $\unsatp$ is a certificate that $1$ is in the ideal generated by $\unsatp$, in the form of a polynomial identity
  \begin{equation}
    \label{eq:ns-refutation}
    \sum_{i=1}^m \polyp_i \polyq_i = 1 
  \end{equation}
  where $\polyq_1, \ldots, \polyq_{m}$ are arbitrary polynomials in $\F[x_1, \ldots, x_n]$. 
  
  Nullstellensatz is sound: no certificate \eqref{eq:ns-refutation} can exist if the polynomials in $\unsatp$ have a common root, since assigning it to both sides of \eqref{eq:ns-refutation} results in $0 = 1$. It is well known that Nullstellensatz is complete over any field $\F$ when the variables $x_1, \ldots, x_n$ can only take values in $\{0, 1$\}, meaning that the \emph{Boolean axioms} $\{x_i^2 - x_i : i \in [n]\}$ are present in $\unsatp$.\footnote{More generally, Nullstellensatz is complete over any field whenever the domains of $x_1, \ldots x_n$ are all finite.} The Boolean axioms are present in all sets $\unsatp$ considered in this paper. The primary complexity measure of a Nullstellensatz refutation as in \eqref{eq:ns-refutation} is \emph{degree}, defined as the maximum degree among the polynomials in it.

The proof system considered in this paper is \emph{polynomial calculus}~\cite{CEI96Groebner}, which can be seen as a dynamic version of Nullstellensatz. %
Starting from $\unsatp$, polynomial calculus derives new polynomials in the ideal $\langle \unsatp \rangle$ through two
derivation rules:
\begin{subequations}
\begin{align}
    \label{eq:pc-lin-combn}
  \textit{Linear combination: }
  &\frac{\pcpolyp \quad \pcpolyq}{a\pcpolyp + b\pcpolyq},\; a, b \in \F \, ;\\
    \label{eq:pc-mult}
  \textit{Multiplication: }
  & \frac{\pcpolyp}{\indeterminate_i \pcpolyp},\; \indeterminate_i \text{ any variable.}
\end{align}
\end{subequations}
A \emph{polynomial calculus derivation} of a
polynomial~$\pcpolyp$ from the set~$\unsatp$
is a sequence
of polynomials
$(\pcpolyp_1, \ldots, \pcpolyp_\tau)$, 
where
$\pcpolyp_\tau = \pcpolyp$ and each $\pcpolyp_i$
is either
in~$\unsatp$ or
obtained by applying one of
\eqref{eq:pc-lin-combn}--\eqref{eq:pc-mult}
to polynomials $\pcpolyp_j$ with $j < i$. A \emph{polynomial calculus
refutation of $\unsatp$} is a derivation of
$1$ from $\unsatp$. 

The most common complexity measures
of polynomial calculus refutations are \emph{size} and \emph{degree}.
The \emph{size} of a polynomial $p$ is its number of monomials when
expanded into a linear combination of distinct monomials, and the
\emph{degree} of $\pcpolyp$ is the maximum degree among all
of %
its monomials. The size of a polynomial calculus refutation
$\proofstd$ is the sum of the sizes of the polynomials in~$\proofstd$,
and the degree of $\proofstd$ is the maximum degree among all
polynomials in~$\proofstd$. 

This paper only concerns degree lower bounds for polynomial calculus, but nevertheless we should mention that a strong enough polynomial calculus degree lower bound immediately implies a size lower bound through the \emph{size-degree relation}~\cite{IPS99LowerBounds}: if all polynomials in a set $\unsatp\subseteq \F[x_1, \ldots, x_n]$ have degree at most $\degstd_0$ and any polynomial calculus refutation of $\unsatp$ requires degree at least $\degstd$, then any polynomial calculus refutation of $\unsatp$ requires size $\exp(\Omega((\degstd-\degstd_0)^2/n))$.%

For technical reasons, most papers about polynomial calculus in fact consider the better-behaved proof system \emph{polynomial calculus resolution (PCR)}
\cite{ABRW02SpaceComplexity}, %
where
additionally each variable $x_i$ appearing in~$\unsatp$ has a formal
negation $\overline{x}_i$ enforced by adding the set of polynomials
$\{x_i + \overline{x}_i - 1 : i \in [n]\}$ to $\unsatp$. Polynomial calculus and PCR are
equivalent with respect to degree; %
 hence, we will be lax in differentiating between them in this paper. 

 The most common technique for proving polynomial calculus degree lower bounds is by constructing a \emph{pseudo-reduction operator}, as mentioned in \cref{sec:discussion-of-pf-techniques}.

 \begin{definition}[Pseudo-reduction operator]
  \label{def:r-op}
  Let~$\F$ be a field, let~$\pcdeg \in \N^+$ and let $\unsatp\subseteq \F[x_1, \ldots, x_\nvars]$ be unsatisfiable. An~$\F$-linear
  operator~$\pseudred$ on $\F[x_1, \ldots, x_n]$ 
  is a \emph{degree-$\pcdeg$ pseudo-reduction operator for~$\unsatp$} if
  \begin{enumerate}
  \item~$\pseudred(1) = 1$, \label{thesis:eq:rprop1}%
    \label{item:rop-property-1}
  \item~$\pseudred(\polyp) = 0$ for every polynomial
   ~$\polyp\in\unsatp$,
    \label{item:rop-property-2}
  \item
  ~$\pseudred(x_im) =
    \pseudred\bigl(x_i
    \pseudred(m)\bigr)$ for every monomial~$m$ of degree at most
  $\degstd-1$ and every variable~$x_i$.
    \label{item:rop-property-3}
  \end{enumerate}
\end{definition}
The kernel of a pseudo-reduction operator provides an overapproximation of what can be derived from $\unsatp$ in degree at most~$\pcdeg$ which is fine-grained enough to not contain~$1$. The existence of a degree-$\degstd$ pseudo-reduction operator for a set~$\unsatp$ of polynomials implies a polynomial calculus degree lower bound of $\pcdeg+1$ for refuting~$\unsatp$, as stated in \cref{lem:r-operator}.\footnote{It is folklore that the converse direction also holds, and hence that degree-$\degstd$ pseudo-reduction operators for $\unsatp$ characterize the existence of degree-$\degstd$ polynomial calculus refutations of $\unsatp$. However, to the best of our knowledge, this result has never been published.}%
\begin{lemma}[\cite{Razborov98LowerBound}, Lemma~3.2]
  \label{lem:r-operator}
  Let~$\pcdeg \in \N^+$ and let~$\unsatp$ be a set of polynomials over
 ~$\F[x_1, \ldots, x_\nvars]$. If there exists a degree-$\pcdeg$
  pseudo-reduction operator for~$\unsatp$, then any polynomial calculus
  refutation of~$\unsatp$ over~$\F$ requires degree strictly greater
  than~$\pcdeg$.
\end{lemma}
\begin{proof}[Proof sketch]
Apply
$\pseudred$ to all polynomials in a purported polynomial
calculus refutation of~$\unsatp$ of degree at most~$\pcdeg$. Conclude by induction on refutation length that it is
impossible to reach the constant polynomial~$1$.
\end{proof}

\subsection{Polynomial Encoding of Homomorphisms}
\label{sec:polynomial-encoding-csp}
Given an instance $\rela$ of $\csp{\relt}$, we encode the statement that $\rela \to \relt$ as a set of polynomials $\polyhomo{\rela}{\relt}$ whose common roots correspond precisely to the homomorphisms $\rela \to \relt$. We use Boolean variables $\{x_{a, i}: \verta \in \ala, \verti \in \alt\}$ to indicate that $\verta \in \ala$ is mapped to $\verti \in \alt$. For a field $\F$, we denote the polynomial ring over $\F$ in these variables by $\F[\mathbf{x}_{\ala, \alt}]$. A (partial) function $f \colon A \to T$ is thus encoded as a monomial $m_f = \prod_{(\verta, f(a))} x_{\verta, f(a)}$. For a monomial $m\in \F[\mathbf{x}_{\ala, \alt}]$, we denote the set of vertices in $\ala$ that appear in some variable of $m$ by $\verts{\ala}(m)$. The set $\polyhomo{\rela}{\relt}$ consists of the polynomials in \eqref{eq:boolean-encoding1}--\eqref{eq:boolean-encoding4}. 
    \begin{subequations}
        \begin{align}
          \label{eq:boolean-encoding1}
          \smash{\sum_{\verti\in \alt} x_{\verta, \verti} - 1}
          \qquad &  \verta \in \ala
          &&[\text{every vertex $\verta \in \ala$ mapped to some $\verti\in \alt$}]
          \\
          \label{eq:boolean-encoding2}
          x_{\verta, \verti}x_{\verta, \verti'}
          \qquad &  \verta \in \ala,\ \verti \neq \verti' \in \alt
          &&[\text{no $\verta \in \ala$ mapped to $>1$ vertex $\verti \in \alt$}] 
          \\
          \label{eq:boolean-encoding3}
          \prod_{\verta \in \mathbf{a}} x_{\verta, f(\verta)} 
          \qquad & \begin{array}{l}
            \hspace{-5pt} 
            \relsymbol \in \sigma, \ \mathbf{a} \in \relsymbol^{\rela},\\ 
            \hspace{-5pt} 
            f\colon \mathbf{a} \mapsto f(\mathbf{a}) \in \alt^{\ar(\relsymbol)} \setminus \relsymbol^{\relt} \end{array}
          &&[\text{no tuple in $\relsymbol^{\rela}$ mapped to tuple not in $\relsymbol^{\relt}$}] 
          \\
          \label{eq:boolean-encoding4}
          x_{\verta, \verti}^2 - x_{\verta, \verti}
          \qquad & \verta \in \ala, \ \verti \in \alt
          &&[\text{Boolean axioms}]
        \end{align}   
      \end{subequations}

\subsection{CSP Hierarchies}
\label{sec:csp-hierarchies}
In this section we define the CSP hierarchy algorithms that we consider, namely the \emph{$\kcons$\nobreakdash-consistency} algorithm and its \emph{$\Z$-affine} and \emph{cohomological} strengthenings.

\paragraph{The \textit{k}-Consistency Algorithm:} 
Given an instance $\rela$ of $\csp{\relt}$, the well-known \emph{$\kcons$\nobreakdash-consistency algorithm} accepts $\rela$ if and only if there exists a nonempty collection 
    $\kappa =\set{\kappa(X)\mid X\in\binom{A}{\leq \kcons}}$ 
    which contains, for each \smash{$X\in\binom{\ala}{\leq \kcons}$}, a set $\kappa(X)$ of partial homomorphisms $\rela[X]\to \relt$ such that the following conditions are satisfied for all $Y\subseteq X \in \smash{\binom{\ala}{\leq \kcons}}$:
\begin{description}
    \item \emph{Extendability:} 
    Each homomorphism $\homomphi \in \kappa(Y)$ extends to some $\homompsi\in \kappa(X)$.\footnote{This condition is called \emph{forth} by some authors in constraint satisfaction, and \emph{flasqueness} or \emph{flabbiness} in sheaf terminology.}
    \item \emph{Down-closure:} For every homomorphism $\homomphi\in \kappa(X)$, it holds that $\restrict{\homomphi}{Y} \in \kappa(Y)$.  
\end{description}
A collection $\kappa$ satisfying the extendability and down-closure conditions is called \emph{$\kcons$\nobreakdash-consistent}. Since a union of $\kcons$\nobreakdash-consistent collections is again $\kcons$\nobreakdash-consistent, there is a unique maximal $\kcons$\nobreakdash-consistent collection in any collection of partial homomorphisms. 

\paragraph{$\Z$-Affine \textit{k}-Consistency:}
The \emph{$\Z$-affine $\kcons$\nobreakdash-consistency algorithm} \cite[Section~4.3]{DO19LocalConsistency} is a combination of the $\kcons$\nobreakdash-consistency algorithm and %
an integer affine relaxation
which, for an instance $\rela$ of $\csp{\relt}$, proceeds as follows. First, the $\kcons$\nobreakdash-consistency algorithm is run to obtain a $\kcons$\nobreakdash-consistent collection $\kappa = \set{\kappa(X) \mid X\in \smash{\binom{A}{\leq \kcons}}}$ of partial homomorphisms. Then, the system $L_\kcons(\rela, \relt, \kappa)$ of linear equations in \eqref{eq:zaff-base}--\eqref{eq:zaff-consistency} is solved over $\Z$. The system $L_\kcons(\rela, \relt, \kappa)$ is over the variables {$x_{X, \homomphi} $}, for all {$X\in \smash{\binom{\ala}{\leq \kcons}}$} and all {$\homomphi \in \hom(\rela[X], \relt)$}. 
    \begin{subequations}
        \begin{align}
            \label{eq:zaff-base}
            \sum_{
                \substack{
                    \homomphi \in \kappa(X)
            }}
            x_{X, \homomphi} &= 1
             &&\text{for all $X\in \binom{A}{\leq \kcons}$} \\
            \label{eq:zaff-consistency}
            \sum_{\homomphi \in \kappa(X), \restrict{\homomphi}{Y} = \homompsi}x_{X, \homomphi} &= x_{Y, \homompsi} 
            &&\text{for all $Y \subseteq X\in \binom{A}{\leq k}$ and $\homompsi \in \kappa(Y)$}
        \end{align}
    \end{subequations}
The $\Z$-affine $\kcons$\nobreakdash-consistency algorithm accepts $\rela$ if and only if there exists a solution to $L_\kcons(\rela, \relt)$ for some $\kcons$\nobreakdash-consistent set $\kappa$ (without loss of generality, the maximal such $\kappa$). %
 Note that if $\kcons$ is at least the maximal arity among the relations in $\sigma$, then the equations \eqref{eq:zaff-base}--\eqref{eq:zaff-consistency} have a $\{0, 1\}$-valued solution over the full set of partial homomorphisms of domain size at most $\kcons$ (which is not necessarily $\kcons$\nobreakdash-consistent) if and only if $\rela \to \relt$.

\paragraph{Cohomological \textit{k}-Consistency:}
The \emph{cohomological $k$-consistency} algorithm due to \'O~Conghaile~\cite{OConghaile22Cohomology} further strengthens the $\Z$-affine $k$-consistency algorithm. The algorithm was originally formulated in sheaf-theoretic language but can also be defined in more elementary terms as follows. %
    \begin{enumerate}
      \item[]\hspace*{-\leftmargin}\textit{Cohomological $k$-consistency algorithm for $\csp{\relt}$ on an instance $\rela$}:
    \item Maintain a collection $\mathcal H=\set{\mathcal H(X)\mid X\in \smash{\binom{A}{\leq k}}}$ of partial homomorphisms $\rela\to\relt$. 
    For each $X\in \binom{A}{\leq k}$, initialize $\mathcal{H}(X) = \hom(\rela[X], \relt)$.
    \item Repeat until each set $\mathcal{H}(X)$ stabilizes:
    \begin{enumerate}
    \item %
    Find the maximal $\kcons$\nobreakdash-consistent collection $\kappa$ in $\mathcal H$.  Update $\mathcal H \leftarrow \kappa$. %
    \item For each $X \in \binom{A}{\leq k}$ and $\homomphi \in \mathcal{H}(X)$, check whether $L_k(\rela, \relt, \mathcal{H})$ has a solution satisfying $x_{X, \homomphi} = 1$ and $x_{X, \homompsi} = 0$ for every $\homompsi\in \mathcal{H}(X)\setminus \{\homomphi\}$. If not, remove $\homomphi$ from $\mathcal{H}(X)$ for the next iteration of the loop.
    \end{enumerate} 
    \item If $\mathcal{H}(X) = \varnothing$ for some $X\in \binom{A}{\leq k}$, then reject; otherwise accept.  
    \end{enumerate}

The $\kcons$\nobreakdash-consistency algorithm always correctly decides whether $\rela \to \relt$ when $\kcons= \setsize{\ala}$ (but then does not necessarily run in time polynomial in $\rela$). Hence, the same is true for the $\Z$-affine and cohomological $\kcons$\nobreakdash-consistency algorithms, and a level lower bound of $\Omega(\setsize{\ala})$ is therefore asymptotically optimal for all three hierarchies.

It is not hard to see that all algorithms for $\csp{\relt}$ in this section always accept $\rela$ if $\rela \to \relt$. To prove that these algorithms do not solve $\csp{\relt}$, we should therefore exhibit a \emph{fooling instance}: a structure $\rela$ such that $\rela \not \to \relt$, but such that the algorithms still accept $\rela$. To show that a hierarchy does not solve a PCSP with template $(\rels, \relt)$, we should find an instance $\rela$ such that $\rela \not\to \relt$, but such that the hierarchy for $\csp{\rels}$ accepts $\rela$. This setting is highly reminiscent of that in proof complexity, and in Sections \ref{sec:csp-hiearchies-ipr-operators} and \ref{sec:ar-method-cohomological-consistency} we show how techniques from proof complexity can be leveraged to construct fooling instances.

\section{CSP Hierarchies and Pseudo-Reduction Operators}
\label{sec:csp-hiearchies-ipr-operators}
In this section, we show that the existence of a pseudo-reduction operator satisfying a certain integrality condition implies level lower bounds for the %
$\Z$-affine $\kcons$\nobreakdash-consistency hierarchy. In \cref{sec:ar-method-cohomological-consistency}, we provide a method for constructing such pseudo-reduction operators; this construction also lets us extend our lower bounds to hold for the cohomological $\kcons$\nobreakdash-consistency algorithm. 

We are particularly interested in pseudo-reduction operators $\pseudred$ for $\polyhomo{\rela}{\relt}$ on $\F[\mathbf{x}_{\ala, \alt}]$ where $\F$ is of characteristic zero and hence contains $\Z$ as a subring. If restricting the domain of $\pseudred$ to $\Z[\mathbf{x}_{\ala, \alt}]$ results in an operator $\Z[\mathbf{x}_{\ala, \alt}] \to \Z[\mathbf{x}_{\ala, \alt}]$, which by linearity is equivalent to $\pseudred$ mapping all monomials to polynomials in $\Z[\mathbf{x}_{\ala, \alt}]$, we say that $\pseudred$ is an \emph{integral} pseudo-reduction operator, or \emph{IPR} operator for short. For a degree-$\degstd$ pseudo-reduction operator  $\pseudred$ for $\polyhomo{\rela}{\relt}$, the \emph{$\pseudred$-support} is the collection $\kappa_\pseudred = \{\kappa_\pseudred(X): {X\in \binom{\ala}{\leq \kcons}}\}$ where $\kappa_\pseudred(X)= \{\homomphi \in \hom(\rela[X], \relt) \mid \pseudred(m_\homomphi) \neq 0\}$.

The main result of this section is \cref{lem:accepted-by-z-aff-cons}.

\begin{lemma}
    \label{lem:accepted-by-z-aff-cons}
    If $\pseudred$ is a degree-$\degstd$ pseudo-reduction operator for $\polyhomo{\rela}{\relt}$, then $\kappa_\pseudred$ is $\kcons$\nobreakdash-consistent for all $\kcons\leq \degstd$. Moreover, if $\pseudred$ is integral, then for any linear $h: \Z[\mathbf{x}_{\ala, \alt}]\to \Z$ satisfying $h(1)=1$, the assignment to the variables $x_{X, \homomphi}$ of \eqref{eq:zaff-base}--\eqref{eq:zaff-consistency} which sets $x_{X, \homomphi}$ to $h \circ R(m_\homomphi)$ for $X \in \binom{\ala}{\leq \degstd}$ and $\homomphi \in \hom(\rela[X], \relt)$ is a solution to \eqref{eq:zaff-base}--\eqref{eq:zaff-consistency} over $\Z$. 
    
    Consequently, if there exists a degree-$D$ IPR operator for $\polyhomo{\rela}{\relt}$, then the $\Z$-affine $\kcons$\nobreakdash-consistency algorithm for~$\csp{\relt}$ accepts $\rela$ for all $\kcons \leq \degstd$. 
\end{lemma}

\begin{proof}
    We first observe that $\kappa_R$ in fact contains the whole support of $\pseudred$ on monomials of degree at most $\degstd$. To see this, note that every monomial $m$ in the variables of $P_{\rela\to\relt}$ encodes a relation $\relsymbol_m \subseteq A \times T$, where $(a, i) \in A\times T$ is contained in $\relsymbol_m$ if and only if $x_{a, i} \in m$, with domain size at most $D$. If $\relsymbol_m$ is not a partial homomorphism $\rela \to \relt$, then $\relsymbol_m$ either maps some $\verta\in \ala$ to multiple $i \in T$ or $\relsymbol_m$ maps a tuple in a relation $\relsymbol^{\rela}$ in $\rela$ to a tuple which is not in $\relsymbol^\relt$. In both cases, we have that $\pseudred(m) = 0$ since $m$ is then a multiple of an axiom in \eqref{eq:boolean-encoding2} or \eqref{eq:boolean-encoding3}. Therefore, the support of $\pseudred$ consists of monomials that encode partial homomorphisms in $A$. 
    
    We first show that the collection $\kappa_R$ is $\kcons$\nobreakdash-consistent for all $\kcons \leq \degstd$ and then show that the output of $\pseudred$ satisfies \eqref{eq:zaff-base}--\eqref{eq:zaff-consistency} \emph{as polynomials}, using only the homomorphisms in $\kappa_R$. Therefore, composing $\pseudred$ with a linear map $h\colon \Z[\mathbf{x}_{\ala, \alt}]\to \Z$ evaluating to $1$ on $1$ yields a solution over $\Z$ to \eqref{eq:zaff-base}--\eqref{eq:zaff-consistency}. 
    
    To prove that $\kappa_R$ is $\degstd$-consistent, we need to show that $\kappa_R$ is non-empty and satisfies the extendability and down-closure conditions for all $\kcons \leq \degstd$. For the down-closure condition we use the contrapositive. Recall that for a (partial) homomorphism $\homomphi \colon \rela \to  \relt$, we denote the monomial $\smash{\prod_{(\verta, \homomphi(\verta))}} x_{\verta, \homomphi(\verta)}$ by $m_\homomphi$. Note that for $Y \subsetneq X \in \smash{\binom{A}{\leq k}}$ and homomorphisms $\homomphi \colon \rela[X] \to \relt$ and $\homompsi \colon \rela[Y] \to \relt$ such that $\restrict{\homomphi}{Y} = \homompsi$, it holds that $m_\homomphi$ is a multiple of $m_\homompsi$. Therefore, if $\pseudred(m_\homompsi) = 0$ then by Property~\ref{item:rop-property-3} of a pseudo-reduction operator we have
    \begin{equation}
    R(m_\homomphi) = R(m'\cdot m_\homompsi) = R(m'\cdot R(m_\homompsi)) = R(0) = 0\eqcomma
    \end{equation}
    so if $\homompsi\not\in \kappa_R(Y)$ then for any $X\supseteq Y$, no extension of $\homompsi$ is is in $\kappa_R(X)$.  

    To establish the extendability condition, let $Y \subsetneq X \in \binom{A}{\leq k}$ where $X = Y \union \{a\}$. For any homomorphism $\homompsi \colon \rela[Y] \to \relt$ in $\kappa_R[Y]$, we have by Property~\ref{item:rop-property-3} of a pseudo-reduction operator that
    \begin{equation}
        R\Big(\Big(\sum_{\verti \in \alt} x_{\verta, \verti} - 1\Big) \cdot m_\homompsi \Big) = R\Big(R\Big(\sum_{\verti \in \alt} x_{\verta, \verti} - 1\Big) \cdot m_\homompsi \Big) = 0 \eqcomma
    \end{equation} 
    and by linearity we also have
    \begin{equation}
        R\Big(\Big(\sum_{i \in \alt} x_{\verta, \verti} - 1\Big)\cdot m_\homompsi\Big)  = \sum_{\verti \in \alt} R(x_{\verta, \verti}\cdot m_\homompsi) - R(m_\homompsi)\eqperiod  
    \end{equation}
    Since $\homompsi \in\kappa_R$ and hence $\pseudred(m_\homompsi) \neq 0$, it follows that at least one of the extensions $x_{a, i}m_\homompsi$ of $\homompsi$ to $X$ is also in $\kappa_R$. With the extendability and down-closure conditions established, we conclude that $\kappa_R$ is $\kcons$\nobreakdash-consistent.
    
    Now, we show that the output of $\pseudred$ satisfies \eqref{eq:zaff-base}--\eqref{eq:zaff-consistency}. For~\eqref{eq:zaff-base}, we use induction on $\setsize{X}$. The base case, where $\setsize{X} = 0$, follows since $R(1)=1$. %
    For the induction step, let $S(X)$ denote the sum $\sum_{\homomphi \in k_R(X)} m_\homomphi$ for $\kcons \leq \degstd$ and suppose that for all sets $Y \in \smash{\binom{A}{\leq \kcons-1}}$ we have that $\pseudred(S(Y)) = 1$. For any set $X = Y \union \{a\}$ of size $\kcons$, note that all summands in $S(X)$ are contained in the sum \smash{$\sum_{i \in T} x_{a, i}\cdot S(Y)$} by $\kcons$\nobreakdash-consistency, and by Property~\ref{item:rop-property-3} of a pseudo-reduction operator we have
    \begin{equation}
        R\Big(\sum_{i \in T} x_{a, i} \cdot (S(Y) - 1)\Big) = R\Big(\sum_{i \in T} x_{a, i} \cdot R(S(Y) - 1)\Big) = 0\eqperiod
    \end{equation}
    Moreover, we have by linearity that
    \begin{subequations}
        \begin{align}
        R\Big(\sum_{i \in T} x_{a, i} \cdot (S(Y) - 1)\Big) 
        &= R\Big(\sum_{i \in T} x_{a, i} \cdot S(Y)\Big) - R\Big(\sum_{i \in T} x_{a, i}\Big) \\
        &= R\Big(\sum_{i \in T} x_{a, i} \cdot S(Y)\Big) - 1 &&[\text{induction hypothesis}]\\
        &= R\Big(S(X)\Big) - 1\eqcomma &&[\text{definition of $\kappa_R$}]
        \end{align}
    \end{subequations}
    which proves that \eqref{eq:zaff-base} is satisfied when composing $\pseudred$ with any integer-valued $\F$-linear map $h$ such that $h(1)=1$. 
    To establish \eqref{eq:zaff-consistency}, first observe that for $Y \subsetneq X \in \binom{A}{\leq k}$ we have that $\pseudred(S(X\setminus Y)) = 1$. Fix a homomorphism $\homompsi: \rela[Y] \to \relt$ in $\kappa_R$. By Property~\ref{item:rop-property-3} of a pseudo-reduction operator, we have
    \begin{equation}
        R\big((S(X\setminus Y)-1)\cdot m_\homompsi\big) = R(R((S(X\setminus Y)-1))\cdot m_\homompsi) = 0
    \end{equation}  
    and moreover %
    \begin{subequations}
        \begin{align}
        0 &= R\big((S(X\setminus Y)-1)\cdot m_\homompsi\big) \\
        &= \Bigl(\sum_{\homomu \in \kappa_R(X\setminus Y)} R(m_\homomu \cdot m_\homompsi)\Bigr) - R(m_\homompsi) \\
        &= \Bigl(\sum_{\homomphi\in \kappa_R(X): \restrict{\homomphi}{Y} = \homompsi}R(m_\homomphi)\Bigr) - R(m_\homompsi) &&[\text{definition of $\kappa_R$}]
        \end{align}
    \end{subequations}
    which establishes \eqref{eq:zaff-consistency}. Hence, again composing $\pseudred$ with any linear map $h: \Z[\mathbf{x}_{\ala, \alt}]\to \Z$ such that $h(1)=1$ furnishes a level-$D$ solution for the $\Z$-affine $k$-consistency algorithm. \qedhere

\end{proof}

\begin{remark}
    \label{rmk:ipr-operators-work-every-field}
    The notion of an IPR operator $\pseudred$ for a set $\unsatp$ is only well-defined over fields of characteristic zero. Note, however, that the existence of such an operator $\pseudred$ over any field of characteristic zero implies a polynomial calculus degree lower bound for refuting $\unsatp$ over \emph{all} fields, even those of finite (prime) characteristic $\ell$. This is because $((\cdot) \mod {\ell}) \circ \pseudred$ is then a pseudo-reduction operator over the finite field $\F_\ell$, and every field of characteristic $\ell$ contains $\F_\ell$ as a subfield.  
\end{remark}

\section{From Pseudo-Reduction To Cohomological \emph{\kcons}-Consistency}
\label{sec:ar-method-cohomological-consistency}
In \cref{sec:csp-hiearchies-ipr-operators}, we showed that the existence of a degree-$\degstd$ IPR operator for $\polyhomo{\rela}{\relt}$ implies that the $\Z$-affine $\kcons$\nobreakdash-consistency algorithm for $\csp{\relt}$ accepts $\rela$ for all $\kcons\leq \degstd$. This section introduces the \emph{Alekhnovich--Razborov method}~\cite{Razborov98LowerBound,AR03LowerBounds} for constructing such operators. The proof that operators produced by this method are integral is deferred to \cref{sec:ipr-operators-lex-game}. Given integrality, this section concludes by showing that a degree-$\degstd$ pseudo-reduction operator for $\polyhomo{\rela}{\relt}$ constructed using the Alekhnovich--Razborov method can be used to prove that $\rela$ is a fooling instance for the cohomological $\kcons$\nobreakdash-consistency algorithm for $\csp{\relt}$; this is formally established in \cref{lem:cohomological-k-consistency-lower-bounds}. %

The treatment of the Alekhnovich--Razborov method in this section partly follows that of \cite[Section~4.2]{CdRNPR25GraphColouring} but is specialized to sets of polynomials $\polyhomo{\rela}{\relt}$ as in \cref{sec:polynomial-encoding-csp}. For more general perspectives, see~\cite{AR03LowerBounds} and the more recent papers \cite{MN24Generalized,CdRNPR25GraphColouring}. 

Recall from \cref{sec:discussion-of-pf-techniques} that if $\rela \to \relt$, %
then the true reduction operator $\redop_{\ideal{\polyhomo{\rela}{\relt}}}$ is a degree\nobreakdash-$\degstd$ pseudo-reduction operator for all $\degstd$. For structures $\rela \not\to \relt$, the Alekhnovich--Razborov method constructs operators that attempt to mimic the behavior of $\redop_{\ideal{\polyhomo{\rela}{\relt}}}$ as closely as possible while mapping $1$ to $1$. The construction is as follows: Let $\prec$ be an admissible ordering on the monomials in $\F[\mathbf{x}_{A, T}]$. To each monomial $m$, associate a \emph{satisfiable} set $S(m) \subseteq \polyhomo{\rela}{\relt}$ and define $\pseudred(m)$ to be $\redop^{\prec}_{\ideal{S(m)}}(m)$. Finally, define $\pseudred$ on arbitrary polynomials by linear extension. 

For a set of polynomials $\polyhomo{\rela}{\relt}$ and a monomial $m = \prod x_{\verta, \verti}$ in $\F[\mathbf{x}_{A, T}]$, a natural choice for $S(m)$ is a set $\polyhomo{\rela[X_{m}]}{\relt} \subseteq \polyhomo{\rela}{\relt}$, where $X_m\subseteq \ala$ depends on $m$ and $\rela[X_m]$ is a satisfiable substructure of $\rela$. In addition, we require that $X_m$ only depends on the vertices of $\ala$ mentioned by $m$. To obtain lower bounds for the cohomological $\kcons$\nobreakdash-consistency algorithm, it is necessary for our proof that $S(m)$ be of this form; see \cref{lem:cohomological-k-consistency-lower-bounds}. Therefore, we only consider such sets $\sop(m)$ and again refer the reader to~\cite[Section~4.2]{CdRNPR25GraphColouring} for a more general treatment. Given a set $X\subseteq \ala$, we denote the ideal $\ideal{\polyhomo{\rela[X]}{\relt}}$ by $\ideal{X}$ for brevity when the relational structures $\rela$ and $\relt$ are clear from context. 

As previously noted, the substructure $\rela[X_{m}]$ should be the part of $\rela$ that is ``relevant'' to $m$. The next definition takes the first step toward formalizing this intuition by putting some reasonable constraints on $\sop$, which are precisely those of the well-known notion of a \emph{closure operator}. A closure operator takes each set $X\subseteq A$ to its \emph{closure} $\cl{X}$, which the smallest ``closed'' set containing $X$. Here, a closed set $C\subseteq \ala$ with respect to $\rela$ and $\relt$ is one where the substructure $\rela[C]$ contains the part of $\rela$ that is ``relevant'' for reducing monomials that mention only vertices in $C$ modulo $\ideal{\polyhomo{\rela}{\relt}}$, in a sense we make precise in \cref{lem:reducibility-gives-reduction-operator}. In general, the definition of a closed set varies with the application; see Sections \ref{sec:approximate-graph-coloring} and \ref{sec:lax-and-null-constraining} for two examples.
\begin{definition}[Size-$\degstd$ closure operator]
    Given a set $A$, a map $\operatorname{cl} \colon 2^{\ala} \to 2^{\ala}$ is a \emph{size-$\degstd$ closure operator} if the following conditions hold for all sets $X, Y \subseteq A$ of size at most $\degstd$.  
    \begin{description}
        \item Self-containment: $X \subseteq \cl{X}$. 
        \label{item:monotonicity}
        \item Monotonicity: If $X\subseteq Y$, then $\cl{X} \subseteq \cl{Y}$.
        \label{item:idempotence} 
        \item Idempotence: $\cl{\cl{X}} = \cl{X}$.  
    \end{description}
    If these conditions hold for $\degstd = \setsize{\ala}$, we say that $\operatorname{cl}$ is a \emph{closure operator.}
\end{definition}

Given a closure operator $\cl{\cdot}$ on $\ala$, we define the closure of a monomial $m\in \F[\mathbf{x}_{\ala, \alt}]$ to be the set $\cl{m} = \cl{\verts{\ala}(m)} \subseteq \ala$. The next lemma formalizes that the closure of a monomial should be the part of $\ala$ such that $\rela[\cl{m}]$ contains all ``relevant'' information about reducing $m$ modulo $\ideal{\polyhomo{\rela}{\relt}}$, in the sense that adding a small set to $\cl{m}$ does not change the result of reduction. We formulate this property as the so-called \emph{reducibility condition}, following \cite{CdRNPR25GraphColouring}.

\begin{lemma}
    \label{lem:reducibility-gives-reduction-operator}
    Let $\degstd \in \N$, let $\sigma$ be a signature whose relation symbols all have arity at most $\degstd$, let $\relt$ be a $\sigma$-structure and $\rela$ an instance of $\csp{\relt}$, and let $\operatorname{cl} \colon 2^\ala \to 2^\ala$ be a size-$\degstd$ closure operator for $\ala$. Finally, let $\prec$ be an admissible order of the monomials in $\F[\mathbf{x}_{A, T}]$. Suppose that the following conditions hold for every monomial $m \in \F[\mathbf{x}_{\ala, \alt}]$ of degree at most $\degstd$:
    \begin{description}
        \item \emph{Satisfiability:} It holds that $\rela[\cl{m}] \to \relt$.
        \item \emph{Reducibility:} For every monomial $m'$ such that $\cl{m'} \subseteq \cl{m}$, it holds that $m'$ is reducible modulo $\ideal{\cl{m'}}$ if and only if $m'$ is reducible modulo $\ideal{\cl{m}}$.
    \end{description}
    Then, the $\F$-linear extension $\pseudred$ of the map $m  \mapsto \redop^\prec_{\ideal{\cl{m}}}(m)$ is a degree-$\degstd$ pseudo-reduction operator for $\polyhomo{\rela}{\relt}$. 
\end{lemma}
The proof of \cref{lem:reducibility-gives-reduction-operator} can be found in \cref{app:reducibility-gives-reduction-operator}, and is by a now-standard argument first introduced by Alekhnovich and Razborov~\cite{AR03LowerBounds} and subsequently clarified by Galesi and Lauria \cite{GL10Automatizability, GL10Optimality} as well as Mik\v{s}a and Nordstr\"{o}m~\cite{MN24Generalized}. 

In \cref{sec:ipr-operators-lex-game}, we prove that operators $\F[\mathbf{x}_{\ala, \alt}] \to \F[\mathbf{x}_{\ala, \alt}]$ defined as in \cref{lem:reducibility-gives-reduction-operator} with respect to a lexicographic order restrict to operators $\Z[\mathbf{x}_{\ala, \alt}] \to \Z[\mathbf{x}_{\ala, \alt}]$. We record this fact as \cref{lem:ipr-alekhnovich-razborov} and defer its proof to after we prove \cref{lem:integer-reduction}.

\begin{lemma}[Integrality assumption]
    \label{lem:ipr-alekhnovich-razborov}
    Let $\rela, \relt, \cl{\cdot}$, and $\degstd$ be defined as in \cref{lem:reducibility-gives-reduction-operator}. Let $\F$ be any field of characteristic $0$ and let $\prec_{\mathrm{lex}}$ be any lexicographic order on the monomials in $\F[\mathbf{x}_{\ala, \alt}]$. For any monomial $m$ of degree at most $\degstd$, it holds that $\redop^{\prec_{\mathrm{lex}}}_{\ideal{\cl{m}}}(m) \in \Z[\mathbf{x}_{\ala, \alt}]$.
\end{lemma}

Given \cref{lem:ipr-alekhnovich-razborov}, we can establish our main technical lemma needed to prove lower bounds for the cohomological $\kcons$\nobreakdash-consistency algorithm.

\begin{lemma}
    \label{lem:cohomological-k-consistency-lower-bounds}
    Suppose that the assumptions of \cref{lem:reducibility-gives-reduction-operator} are satisfied for all monomials $m\in \F[\mathbf{x}_{A, T}]$ of degree at most $\degstd$, where $\F$ has characteristic zero and monomials are ordered by a lexicographic order $\prec_{\mathrm{lex}}$. Then, the cohomological $\kcons$\nobreakdash-consistency algorithm for $\csp{\relt}$ accepts $\rela$ for all $\kcons \leq \clsize$, where $\clsize$ is the largest integer such that $\setsize{\cl{m}} \leq \degstd$ for all monomials $m$ of degree at most $\clsize$. 
\end{lemma}
\begin{proof}
    By \cref{lem:ipr-alekhnovich-razborov}, the operator $\pseudred \colon \F[\mathbf{x}_{A, T}] \to \F[\mathbf{x}_{A, T}]$ in the statement of \cref{lem:reducibility-gives-reduction-operator} restricts to an operator $\Z[\mathbf{x}_{A, T}] \to \Z[\mathbf{x}_{A, T}]$, and by \cref{lem:reducibility-gives-reduction-operator}, we have that $\pseudred$ is a degree\nobreakdash-$\degstd$ pseudo\nobreakdash-reduction operator for $\polyhomo{\rela}{\relt}$. 
    
    By \cref{lem:accepted-by-z-aff-cons}, the collection $\kappa_\pseudred = \{\kappa_\pseudred(X):{X\in \binom{\ala}{\leq \kcons}}\}$ where $\kappa_\pseudred(X)= \{\homomphi \in \hom(\rela[X], \relt) \mid \pseudred(m_\homomphi) \neq 0\}$ is $\kcons$\nobreakdash-consistent for all $\kcons\leq \degstd$. Moreover, also by \cref{lem:accepted-by-z-aff-cons} the assignment to the variables $x_{X, \homomphi}$ of \eqref{eq:zaff-base}--\eqref{eq:zaff-consistency} which sets $x_{X, \homomphi}$ to $\operatorname{eval}_{y} \circ R(m_\homomphi)$ for $X \in \smash{\binom{\ala}{\leq \kcons}}$ and $\homomphi \in \hom(\rela[X], \relt)$ is a solution to \eqref{eq:zaff-base}--\eqref{eq:zaff-consistency} over $\Z$ for any point $y \in \Z^{\setsize{\ala} \cdot \setsize{\alt}}$. Therefore, it suffices to show that for all sets $X \in \smash{\binom{\ala}{\leq \clsize}}$ and every homomorphism $\homomphi \in \kappa_R(X)$, there exists a point $y \in \Z^{\setsize{\ala} \cdot \setsize{\alt}}$ such that $\operatorname{eval}_{y} \circ R(m_\homomphi) = 1$ and $\operatorname{eval}_{y} \circ R(m_\homompsi) = 0$ for all homomorphisms $\homompsi \neq \homomphi \in \kappa_R(X)$. 

    Denote the polynomial $\pseudred(m_\homomphi)$ by $r$, and recall that $r = \smash{\redop^{\prec_\mathrm{lex}}_{\ideal{\cl{m_\homomphi}}}(m_\homomphi)}$, which by definition is the unique sum of irreducible terms modulo ${\ideal{\cl{m_\homomphi}} = \ideal{\cl{X}}}$ such that $m_\homomphi = r + q$, where $q \in {\ideal{\cl{X}}}$. Note that $q$ vanishes when evaluated on (the indicator vector of) any homomorphism $\rela[\cl{X}] \to \relt$ and that $m_\homomphi = 1$ when evaluated on $\homomphi$. Since $\verts{\ala}(m_\homomphi) = X\subseteq \cl{X}$, it follows that $r$ evaluates to $1$ on any homomorphism $\rela[\cl{X}] \to \relt$ that extends $\homomphi$. Moreover, for all homomorphisms $\homompsi \neq \homomphi \in \kappa_R(X)$, the monomial $m_{\homompsi}$ evaluates to $0$ on that assignment since $\cl{m_\homomphi} = \cl{m_{\homompsi}}$, and hence $\pseudred(m_{\homompsi})$ also does. Therefore, to conclude the proof it remains to show that for all homomorphisms $\homomphi \colon \rela[X] \to \relt$, there exists a homomorphism $\homomu \colon \rela[\cl{X}] \to \relt$ that extends $\homomphi$. But this follows immediately from the $\degstd$-consistency of $\kappa_\pseudred$, since $\setsize{\cl{X}} \leq \degstd$ by assumption. %
    
    Thus, the cohomological $\kcons$\nobreakdash-consistency algorithm stabilizes on $\kappa_\pseudred$, and hence the cohomological $\kcons$\nobreakdash-consistency algorithm for $\csp{\relt}$ accepts $\rela$ for all $\kcons \leq \clsize$ as desired. \qedhere
\end{proof}

\begin{remark}[Room of choice of solutions]
In the proof of \cref{lem:cohomological-k-consistency-lower-bounds}, there exist many maps $h: \mathbb{F}[x] \to \mathbb{F}$ whose composition with $\pseudred$
fools the cohomological $k$-consistency algorithm. Specifically, given a $\kcons$-small $X$ and a homomorphism $\varphi\colon \rela[X] \to \relt$ where $R(m_\varphi) \neq 0$, we can let $h$ be an arbitrary $\mathbb F$-linear extension of the evaluation map $\operatorname{eval}_\mu\colon\mathbb{F}[\mathbf{x}_{\operatorname{cl}(X), T}] \to \mathbb F$ from the subspace $\mathbb{F}[\mathbf{x}_{\operatorname{cl}(X), T}]$ to $\mathbb{F}[\mathbf{x}_{A, T}]$, 
where $\mu$ is an arbitrary homomorphism on $\rela[\operatorname{cl}(X)]$ extending $\varphi$ (which exists by the $D$-consistency of $\kappa_R$). Then $h \circ R$ provides a $\mathbb Z$-affine $k$-consistency solution satisfying $h \circ R(\varphi) = 1$ and $h \circ R(\psi) = 0$ for all other $\psi\colon X \to T$. 
\end{remark}

\begin{remark}[On requiring a lexicographic order]
    \label{rmk:require-lex-order}
    Insisting on a lexicographic order in Lemmas~\ref{lem:ipr-alekhnovich-razborov} and \ref{lem:cohomological-k-consistency-lower-bounds} might seem like a technical artifact of our proof techniques, but there is in fact reason to believe that the same result does not hold for other orders. For instance, the pseudo-reduction operators in \cite{AR03LowerBounds} can be put on the same form as in \cref{lem:reducibility-gives-reduction-operator}, but are defined with respect to a \emph{degree-respecting} monomial order, which the lexicographic order is not. This requirement turns out to be crucial for their proof that the reducibility condition holds (see \cite[Theorem~3.1]{AR03LowerBounds}). Moreover, the operators in \cite{AR03LowerBounds} cannot be integral in general since they then apply over every field (see \cref{rmk:ipr-operators-work-every-field}); this is impossible, since some of the lower bounds in \cite{AR03LowerBounds} apply to sets of polynomials that are easy to refute over some fields. 
\end{remark}

\section{Integrality of Reductions and the Lex Game}
\label{sec:ipr-operators-lex-game}

In this section we establish an integrality property of polynomial reductions (\cref{lem:integer-reduction}), from which \cref{lem:ipr-alekhnovich-razborov} follows as an easy corollary. This property says that for any finite set $V\subseteq \{0,1\}^n$ and any field $\F$ of characteristic 0, the reduction of a monomial $m\in\F[x_1,\ldots,x_n]$ modulo the ideal $I(V)$ under the lexicographic order is always a polynomial with coefficients in~$\Z$. 

The proof of \cref{lem:integer-reduction} %
uses a game characterization of reducibility under a lexicographic order, namely the \emph{lex game} of Felszeghy, R\'ath, and R\'onyai~\cite{FRR06LexGame}; we use an equivalent formulation from~\cite{FR09SomeMeetingPoints}.

\begin{definition}[Lex game; \cite{FR09SomeMeetingPoints}, Section~3]
    Let $\F$ be a field, $V\subseteq \F^n$ a 
    finite set, and $\alpha = (\alpha_1, \ldots, \alpha_n) \in \N^n$ an $n$\nobreakdash-dimensional vector of natural numbers. The \emph{lex game} $\lex{V}{\alpha}$ is the following deterministic game played between two players named Lea and Stan. Lea and Stan both know $V$ and $\alpha$. If $V= \varnothing$, Lea wins immediately. For a nonempty $V$, the lex game proceeds %
    in rounds.
    In round $i$, Lea picks $\alpha_{n-i+1}$ elements of $\F$, and Stan picks a value $y_{n-i+1}$ which is different from all of Lea's choices in that round;
    if he cannot do so, Lea wins immediately. 
    After $n$ rounds, if Stan's answers $(y_1, \ldots, y_n)$ is a vector in $V$, he wins; otherwise, Lea wins.

\end{definition}

The lex game is finite, deterministic, and cannot end in a draw. Therefore, for any $V$ and $\alpha$, one of the players has a winning strategy for $\lex{V}{\alpha}$. 
The main result of~\cite{FRR06LexGame} is that the existence of a winning strategy for Lea %
characterizes %
whether the
monomial $x^\alpha$ is reducible modulo $I(V)$ under the lexicographic order. %

\begin{theorem}[\cite{FRR06LexGame}, Theorem 2]
    \label{thm:lex-game}
    For any field $\F$,
    finite set $V\subseteq \F^n$, and 
    $\alpha %
    \in \N^n$, the monomial $x^\alpha \in \F[x_1, \ldots, x_n]$ is reducible modulo $I(V)$ in the lexicographic order induced by $x_1 \succ\ldots\succ x_n$ if and only if Lea has a winning strategy in the game $\lex{V}{\alpha}$. 
\end{theorem}

Following \cite{FRR06LexGame}, 
we show the following corollary of the game characterization. 
We note that similar results were also shown in \cite[Section~4]{ARS02Shattering}, \cite[Corollary~5.1.2]{Felszeghy07Groebner}, and \cite[Corollary~3]{FR09SomeMeetingPoints} by arguments different from ours.

\begin{lemma}[%
    Integrality lemma]
    \label{lem:integer-reduction}
    If $\F$ is a field of characteristic zero and $V$ is a subset of $\{0, 1\}^n \subseteq \F^n$, then for any monomial $m = x^\alpha\in \F[x_1, \ldots, x_n]$ it holds that $\redop^{\mathrm{lex}}_{I(V)}(m) \in \Z[x_1, \ldots, x_n]$.
\end{lemma}

Intuitively, given a reducible monomial $m$, we can express Lea's guesses in her winning strategy as a 
set
of polynomials 
$f_1(x), \ldots, f_n(x)$, where $x=(x_1,\ldots,x_n)$ represents Stan's answers and 
each %
$f_i$ is a Boolean-valued function that depends only on $\set{x_{i+1},\ldots,x_n}$. 
We will use these polynomials to rewrite $m$ into smaller terms. 

\begin{proof}
We 
use induction over the lexicographic order of $m=x^\alpha$. 
Without loss of generality, we may assume $m$ is multilinear, as $m = m_{\mathrm{mult}}$ modulo $\ideal{x_1^2-x_1,\ldots,x_n^2-x_n}$ where the latter ideal is $I(\{0,1\}^n)\subseteq I(V)$. 

The base case $m=1$ is immediate: 1 either reduces to 0 or itself, both of which are integers.

For the induction step, %
we denote $I=I(V)$. 
If $m$ is irreducible modulo $I$ then $R_{I}(m)=m$, which trivially has integer coefficients. %
Hence, we may suppose that $m$ is reducible modulo $I$. 
By \cref{thm:lex-game}, in this case Lea has a winning strategy in the game $\lex{V}{\alpha}$. Recall that Lea makes a guess for the value of $x_i$ for each $i$ where $\alpha_i=1$, her guess for $x_i$ depends only on Stan's previous answers, and her guesses are $\set{0,1}$-valued since $V\subseteq\set{0,1}^n$. 
Therefore, we can write her guess for $x_i$ as a polynomial 
\[f_i(x_{i+1},\ldots,x_n)=\sum_{a\in\set{0,1}^{n-i}} c_a m_a\]
where $m_a=\prod_{j: a_j=1}x_j\prod_{j: a_j=0}(1-x_j)$ is the indicator polynomial of ``$x=a$'' and $c_a\in\set{0,1}$ is her guess when Stan answered $(x_{i+1},\ldots,x_n)=a$. Clearly, $f_i\in\Z[x_{i+1},\ldots,x_n]$. 
As in \cite{FRR06LexGame}, we now consider the polynomial
\begin{equation}
        p(x) = \prod_{i:\ \alpha_i = 1} (x_i - f_i)
    \end{equation}
which has integer coefficients. The leading term of each factor $x_i-f_i$ is $x_i$ as $f_i$ depends only on $x_{i+1}, \ldots, x_n$, so the leading term of $p$ is $x^\alpha = m$. Moreover, we claim that $p$ is in the ideal $I$. To see this, notice that for any $a \in V$, we have that $p(a)=0$. That is, since Stan loses the lex game instance corresponding to $\alpha$, Lea can always stop him from answering $a$. This means that there is some $i\in [n]$ such that $f_i(a_{i+1},\ldots, a_n)=a_i$.
It follows that $p$ vanishes on $V$ and so is contained in~$I$. 

    We now complete the induction by considering the reduction of $m-p$. As $p \in I$, we have $\redop_{I}(p)=0$ so $\redop_{I}(m) = \redop_{I}(m-p)$. 
    We have shown that $p\in\Z[x_1,\ldots, x_n]$ and that $p$ has leading term $m$, so $m - p = \sum_i d_i m_i$ where $m_i\prec m$ for all $i$ and all $d_i$ are integers. Therefore, $\redop_{I}(m-p)=\sum_i d_i \redop_{I}(m_i)$ where each $\redop_{I}(m_i)$ is in $\Z[x_1, \ldots, x_n]$ by the induction hypothesis. It follows that $\redop_{I}(m)\in\Z[x_1, \ldots, x_n]$. The proof is complete.
\end{proof}

\cref{lem:ipr-alekhnovich-razborov} follows as a straightforward corollary of \cref{lem:integer-reduction}.

\begin{integrality-assumption}[Restated]
    Let $\rela, \relt, \cl{\cdot}$, and $\degstd$ be defined as in \cref{lem:reducibility-gives-reduction-operator}. Let $\F$ be any field of characteristic $0$ and let $\prec_{\mathrm{lex}}$ be any lexicographic order on the monomials in $\F[\mathbf{x}_{\ala, \alt}]$. For any monomial $m$ of degree at most $\degstd$, it holds that $\redop^{\prec_{\mathrm{lex}}}_{\ideal{\cl{m}}}(m) \in \Z[\mathbf{x}_{\ala, \alt}]$.
\end{integrality-assumption}

\begin{proof}
    Note that $\langle \cl{m} \rangle = I(V_{\hom(\rela[\cl{m}], \relt)})$ as ideals in $\F[x_{\rela[\cl{m}],\relt}]$ by the Boolean Nullstellensatz (\cref{lem:boolean-ns}), where $V_{\hom(\rela[\cl{m}], \relt)}\subseteq \{0, 1\}^{\setsize{\cl{m}}\cdot\setsize{\alt}}$ denotes the set of (indicator vectors of) homomorphisms $\rela[\cl{m}] \to \relt$. 
    Then, by \cref{lem:integer-reduction}, $\redop^{\prec_{\mathrm{lex}}}_{\ideal{\cl{m}}}(m) \in \Z[\mathbf{x}_{\cl{m},\alt}]$ for every monomial $m$ over $\mathbf{x}_{\cl{m},\alt}$ that has degree at most $\degstd$. The same remains true when $m$ is over $\mathbf{x}_{\ala,\alt}$ and reductions are done within the larger polynomial ring $\F[\mathbf{x}_{\ala,\alt}]$. To see this, set all variables outside of $\mathbf{x}_{\cl{m},\alt}$ to 0; this assignment does not touch any variable of any generator of $\ideal{\cl{m}}$, nor of $m$, so the reductions in the two polynomial rings coincide by uniqueness of the representation $m = \redop^{\prec_{\mathrm{lex}}}_{\ideal{\cl{m}}}(m) + q$ for $q\in{\ideal{\cl{m}}}$. \qedhere
\end{proof}

\begin{remark}
For readers familiar with the polynomial calculus literature, we note that the lex game characterization of monomial reducibility broadly generalizes the \textit{pigeon dance} criterion, established in \cite{Razborov98LowerBound} and simplified in \cite{IPS99LowerBounds}, for ideals arising from the \emph{pigeonhole principle.}
\end{remark}

\section{Lower Bounds for Approximate Graph Coloring}
\label{sec:approximate-graph-coloring}
In this section, we use the machinery developed in Sections~\ref{sec:csp-hiearchies-ipr-operators}-\ref{sec:ipr-operators-lex-game} to prove lower bounds for the level needed for the cohomological $\kcons$\nobreakdash-consistency algorithm to solve %
approximate graph coloring, also called $\colcons$~vs.~$\ell$-coloring. In this problem, the task is to decide whether the input graph $G$ is either $\colcons$-colorable or not even $\ell$-colorable, for some $\ell \geq \colcons$. Phrased as a PCSP, approximate graph coloring is $\pcsp{K_\colcons}{K_\ell}$, where $K_a$ denotes the complete graph on $a$ vertices. 

The fooling instances are sampled from the \emph{random regular graph model} $\gnd$, which is the uniform distribution over $\graphdeg$-regular graphs on $n$ vertices.\footnote{Note that a $\graphdeg$-regular graph on $n$ vertices can only exist if $\graphdeg n$ is even. We ignore this technical issue in what follows.} We write $G\sim \gnd$ to denote that $G$ is a graph sampled from $\gnd$. More quantitatively, the main result in this section is \cref{th:cohomological-agc-lbs}. 

\begin{theorem}
    \label{th:cohomological-agc-lbs}
    There is an absolute constant $C$ such that the following holds for natural numbers $n$ and $\graphdeg$ such that $\graphdeg \geq 10$ and $6 \graphdeg^3\leq \log{n}$. With high probability as $n\to \infty$, a graph $G\sim\gnd$ has chromatic number strictly larger than $\graphdeg/4\log{\graphdeg}$, but the cohomological $\kcons$\nobreakdash-consistency algorithm for $3$-colorability accepts $G$ for all $\kcons \leq n\cdot \graphdeg^{-C\graphdeg}$. Consequently, the cohomological $\kcons$\nobreakdash-consistency algorithm does not solve $3$~vs.~${\graphdeg/4\log{\graphdeg}}$-coloring on $n$-vertex graphs, for any $\kcons \leq  n\cdot \graphdeg^{-C\graphdeg}$. 
\end{theorem}

The same result holds for sparse Erd\"{o}s-R\'{e}nyi random graphs (where each edge in the graph appears independently with some fixed probability $p$), with a slighly worse dependence on $d$ in the parameters. The proof is essentially the same as that of \cref{th:cohomological-agc-lbs}.  

The rest of this section is devoted to the proof of \cref{th:cohomological-agc-lbs}. Fortunately, all the technical work has already been done in \cite[Sections~5\nobreakdash-6]{CdRNPR25GraphColouring} (partially building on ideas of \cite{RT24GraphsLargeGirth}) for the purpose of proving degree lower bounds for the polynomial calculus proof system in \cref{sec:proof-complexity}. Hence, \cref{th:cohomological-agc-lbs} is, more or less, an immediate consequence of their degree lower bound combined with \cref{lem:cohomological-k-consistency-lower-bounds}, and all that is needed in this section is to reiterate and rephrase their proof for our setting.

We first recall the well-known fact that the chromatic number of a graph $G\sim \gnd$ is concentrated around $\graphdeg/\log{\graphdeg}$.

\begin{lemma}[\cite{KPGW10Chromatic, CFRR02RandomRegular}]
  \label{lem:random-graph-chromatic-number}
With probability $1-o_n(1)$, a graph $G$ sampled from $\gnd$ where $d \leq n^{0.1}$ has chromatic number strictly larger than~$d/4\log{d}$ and at most $d/\log{d}$.\footnote{The cited references establish much stronger concentration bounds than those stated in \cref{lem:random-graph-chromatic-number}, but these bounds do not meaningfully improve the parameters in our main results.}
\end{lemma}

The crucial property of random graphs that is used to establish both the satisfiability and reducibility conditions is \emph{sparsity}.%

\begin{definition}[Sparsity]
    \label{def:sparse}
  A graph $G=(V,E)$ is
  \mbox{\emph{$(\spsize, \epsilon)$\nobreakdash-sparse}} if every vertex
  set~$U\subseteq V$ of size at most $\spsize$ satisfies
  $\setsize{E(U)} \leq (1+ \epsilon)\setsize{ U }$.
\end{definition}

As is well known, random regular graphs are sparse with high probability and with strong parameters.

\begin{lemma}[See, e.g., \cite{CdRNPR25GraphColouring}, Lemma~2.7]
    \label{lem:random-graph-sparse}
    For $n, d \in \N^+$ and $\epsilon, \delta \in \mathbb{R}^+$ satisfying ${\epsilon \delta = \omega(1/\log{n})}$ and $d^2\leq \epsilon \delta \log{n}$, it holds with high probability that a graph $G$ sampled from $\gnd$ is $(\spsize, \epsilon)$-sparse for $\spsize = (8d)^{-(1+\delta)(1+\epsilon)/\epsilon} n$.
\end{lemma}

All not-too-large subgraphs in a sparse graph are $3$-colorable, as shown in the next lemma. Since the proof of this folklore result is short and instructive, we include it with no claim of originality.

\begin{lemma}[Folklore, see, e.g., \cite{CdRNPR25GraphColouring}, Lemma~2.6]
  \label{lem:sparse-colourable}
  If a graph $G=(V, E)$ is $(\spsize, \epsilon)$\nobreakdash-sparse for
  some $\epsilon < 1/2$, then it holds for every subset $U\subseteq V$
  of size at most $\spsize$ that~$G[U]$ is $3$-colorable.
\end{lemma}
\begin{proof}
    We proceed by induction on $\setsize{U}$. The base case $\setsize{U} = 1$ is
  immediate. For the induction step, assume that the claim holds
  for all sets of size at most $s-1$. Consider a set~$U\subseteq V$ of
  size $a \leq \spsize$. The average degree in~$G[U]$
  is~$2\setsize{E(U)}/a$, which is at most $2(1+\epsilon) < 3$ by sparsity. Hence,
  there exists a vertex~$v\in U$ with degree at most 2 in~$G[U]$. The
  graph~$G{[U\setminus \set{v}]}$ is $3$-colorable by the induction
  hypothesis, and every $3$-coloring witnessing this will leave at
  least one color available for $v$. Hence, every
  $3$-coloring of~${G[U\setminus \set{v}]}$ can be extended to
  $G[U]$. The lemma follows.
\end{proof}

The main technical insight needed to establish the satisfiability and reducibility conditions is an appropriate notion of \emph{closed} set (cf. \cref{sec:ar-method-cohomological-consistency}), for which we need the following definitions from \cite[Section~5]{CdRNPR25GraphColouring}. Given a linear order $\prec$ on $V$, an
\emph{increasing (decreasing) path} in~$G$ is a simple path
$(v_1, \ldots, v_\tau)$ where $v_{i} \prec v_{i+1}$
($v_{i} \succ v_{i+1}$) for all $i\in [\tau-1]$. For vertices
$u, v\in V$ we say that~$v$ is a \emph{descendant of $u$} if there
exists a decreasing path from $u$ to $v$, and for a set
$U \subseteq V$, a vertex $v$ is a \emph{descendant of $U$} if it is
a descendant of some vertex in $U$. We denote the
set of all descendants of~$U$ by $\Desc{U}$ and define every vertex to be a
descendant of itself so that $U \subseteq \Desc{U}$.

A \emph{$\hoplength$\nobreakdash-hop with respect to a set $U \subseteq V$} is a
simple path or a simple cycle of length $\hoplength$ with the property
that its two endpoints are both contained in~$U$ (in the case of
cycles, the two endpoints coincide), while all other vertices are not
in~$U$. A \emph{lasso with respect to $U$} is a
walk $(v_1, v_2, v_3, v_4, v_5)$ where $v_2 = v_5$ and all other
vertices are distinct, and where $v_1 \in U$.

\begin{definition}[Closure]
  \label{def:closure}
  Let $G=(V,E)$ be a graph and let $U \subseteq V$. We say that $U$ is
  \emph{closed} if $U = \Desc{U}$ and there are no 2\nobreakdash-,
  3\nobreakdash-, 4\nobreakdash-hops or lassos with respect to $U$. A
  \emph{closure} of $U$ is any minimal closed set that contains~$U$.
\end{definition}

It is shown in \cite[Proposition~5.2]{CdRNPR25GraphColouring} that there is in fact a unique closure for every set $U\subseteq V$, which we denote by $\cl{U}$. The map $\cl{\cdot} \colon 2^V \to 2^V$ satisfies self-containment by definition, and monotonicity and idempotence are straightforward consequences of minimality. Hence, $\cl{\cdot}$ is a closure operator. 

The following lemma shows that the closure of a small enough set $U\subseteq V$ is not much larger than $U$, given a bound on the sizes of descendant graphs.

\begin{lemma}[\cite{CdRNPR25GraphColouring}, Lemma~5.4]%
  \label{lem:closure-size-lemma}
  Suppose that $G=(V, E)$ has a linear order on $V$ and is
  $(\ell, 1/3\sparseparam)$\nobreakdash-sparse for
  $\sparseparam\geq 2$.
  Let~$U\subseteq V$ be a set of size
  $\setsize{U} \leq \ell/25\sparseparam$ such that any decreasing path
  in $G[V\setminus U]$ has at most $\sparseparam$ vertices.
  Then there exists a set $U \subseteq Z\subseteq V$ such that $\setsize{Z} \leq 25 \setsize{U}$ and
  $\cl{U} = \Desc{Z}$.
\end{lemma}

Finally, the next lemma establishes the reducibility condition for $(\spsize, \epsilon)$-sparse graphs of bounded maximum degree. Recall that $\mathbf{x}_{V, [3]}$ denotes the variables in the encoding \eqref{eq:boolean-encoding1}-\eqref{eq:boolean-encoding4} of $3$\nobreakdash-colorability. Given a linear order $\prec_V$ on the vertices of a graph $G$, we say that an order $\prec$ on the monomials in $\F[\mathbf{x}_{V, [3]}]$ \emph{respects} $\prec_V$ if for all $u, v \in V$ and all $i, j \in [3]$, it holds that $x_{u, i} \prec x_{v, j}$ whenever $u \prec_V v$.

\begin{lemma}[\cite{CdRNPR25GraphColouring}, Lemma~6.5]
    \label{lem:local-reduction-strong}
  Let $G=(V, E)$ be a $(\spsize, 1/3\Delta)$\nobreakdash-sparse graph
  with a linear ordering $\prec_V$ on $V$. Let $\F$ be any field and consider an admissible
  order on $\F[\mathbf{x}_{V, [3]}]$ that respects~$\prec_V$. For all sets
  $W\subseteq U$ satisfying that (1) $W$ is closed, (2) 
  $\setsize{U} \leq \ell$, and (3) every vertex in
  $\nbhstd{U\setminus W}{W}$ has degree at most $\Delta$ in
  $G[U\setminus W]$, it holds for every monomial $m$ such that
  $\verts{V}(m) \subseteq W$ that $m$ is reducible modulo
  $\langle U\rangle$ if and only if $m$ is reducible modulo
  $\langle W \rangle$.
\end{lemma}

We now have all components needed to prove \cref{th:cohomological-agc-lbs}. 
\begin{proof}[Proof of \cref{th:cohomological-agc-lbs}]
Let $G\sim \gnd$ and fix a field $\F$ of characteristic $0$. By Lemmas \ref{lem:random-graph-sparse} and \ref{lem:random-graph-chromatic-number}, it holds with high probability that $G$ is $(\spsize, 1/3\graphdeg)$\nobreakdash-sparse for $s = (8\graphdeg)^{-6\graphdeg}\cdot n$ and that $\chi(G) \in [d/4\log{d}, d/\log{d}]$. Given these properties, we establish the satisfiability and reducibility conditions for all monomials of degree at most $\degstd$, where $\degstd = \spsize / 25 \graphdeg^{\chi(G)} = \graphdeg^{-C_1\graphdeg}\cdot n$, for some large enough constant $C_1$. 

We first put a linear order $\prec_V$ on the vertices of $G$ as follows. %
 Let $\chi \colon V \to [\chi(G)]$ be a proper coloring of $G$. Order the vertices of $G$ by color classes, so that $u \prec_V v$ whenever $\chi(u) < \chi(v)$ for all $u, v \in V$, and order the vertices in the same color class arbitrarily. Finally, let $\prec$ be the lexicographic order induced by any linear order on the variables of $\col{G, 3}$ such that $x_{u, i} \prec x_{v, j}$ whenever $u \prec_V v$. Clearly, this lexicographic order respects $\prec_V$, and from now on the closure and reduction operators are defined with respect to this order.

Now we establish the satisfiability condition. Note that with the ordering $\prec_V$, every decreasing path in $G$ has at most $\chi(G)$ vertices, and hence $\setsize{\Desc{U}} \leq \graphdeg^{\chi(G)-1}\setsize{U}$ for all $U \subseteq V$. Therefore, since $\chi(G) \graphdeg^{\chi(G)-1} < \graphdeg^{\chi(G)}$ we have by \cref{lem:closure-size-lemma} that for all sets $U \subseteq V$ of size at most $\spsize / 25 \graphdeg^{\chi(G)}$ that $\setsize{\cl{U}} \leq 25\graphdeg^{\chi(G)}\setsize{U} \leq \spsize$, and hence $G[\cl{U}]$ is $3$-colorable by \cref{lem:sparse-colourable}. The satisfiability condition is thus established for all $\degstd \leq \spsize / 25 \graphdeg^{\chi(G)}$. 

Finally, we turn to the reducibility condition. Let $m$ be a monomial of degree at most $\degstd$ and let $m' \in \F[\mathbf{x}_{V, [3]}]$ satisfy that $\cl{m'} \subseteq \cl{m}$. The argument in the previous paragraph shows that $\setsize{\cl{m}}\leq \spsize$, and therefore the reducibility condition follows immediately from \cref{lem:local-reduction-strong} with $W = \cl{m'}$ and $U = \cl{m}$. 

With the satisfiability and reducibility conditions in place, \cref{lem:cohomological-k-consistency-lower-bounds} implies that the cohomological $\kcons$\nobreakdash-consistency algorithm for $3$-colorability accepts $G$ for all $\kcons\leq \degstd/25\graphdeg^{\chi(G)} \leq \graphdeg^{-C_2\graphdeg}\cdot n$, for a large enough constant $C_2$. The proof is complete.
\end{proof}

\section{Hardness of Lax and Null-Constraining CSPs}
\label{sec:lax-and-null-constraining}
All the technical work of establishing the satisfiability and reducibility conditions for approximate graph coloring was already accomplished in \cite{CdRNPR25GraphColouring}. This section shows the techniques of~\cref{sec:ar-method-cohomological-consistency} ``in action'' by providing a different proof of %
\cite[Theorem~1.3]{CN25Hierarchies}. 
There, Chan and Ng establish that no $n$\nobreakdash-variable CSP whose template $\relt$ satisfies two technical conditions they call \emph{lax} and \emph{null-constraining} is solved by the cohomological $\kcons$\nobreakdash-consistency algorithm for any $\kcons \leq \Omega(n)$. Just as in \cref{sec:approximate-graph-coloring}, the fooling instances for $\csp{\relt}$ are randomly sampled.

The main result of this section is \cref{th:cohomological-lax-nullcons-lbs}. The overall proof structure is the same as in \cref{sec:approximate-graph-coloring}. We first recall a few definitions in order to introduce the CSP instances and closure in \cite{CN25Hierarchies}. Then, we show that the closure of a small set is small (\cref{cor:closure-size}) and that the pseudo-reduction operator based on it is idempotent (\cref{lem:small-set-idempotence}). Finally, we establish the reducibility condition (\cref{lem:reducibility-nullcons}).\footnote{For readers familiar with polynomial calculus, we note that the proof of the reducibility condition %
here is based on \emph{variable restrictions}, which is simpler than its counterpart in \cref{sec:approximate-graph-coloring} and similar to that in~\cite{MN24Generalized}.} 
Theorem \ref{th:cohomological-lax-nullcons-lbs} then follows via the connection to CSP hierarchies in \cref{sec:ar-method-cohomological-consistency}. 
Our argument simplifies that in \cite{CN25Hierarchies} by avoiding their use of \emph{insular families} and %
\emph{augmented closures}~\cite[Section 6]{CN25Hierarchies}, as well as their framework for combined hierarchies~\cite[Section~8]{CN25Hierarchies}. %

In this section we consider only \emph{$\uniform$-uniform} relational structures, where all relations have arity~$\uniform$. For a $\uniform$-uniform $\sigma$-structure $\relt$, an instance $\rela$ of $\csp{\relt}$ induces a $\uniform$-uniform hypergraph $H=(\ala, \hypedgeset)$, where $\hypedgeset = \{\{\verta_1, \ldots, \verta_\uniform\} : \exists \relsymbol\in \sigma 
\mid (\verta_1, \ldots, \verta_\uniform) \in \relsymbol^\rela\}$.

\begin{definition}[Lax; \cite{CN25Hierarchies}, Definition~4.9]
    A $\uniform$-uniform $\sigma$-structure $\relt$ is \emph{lax} if for every symbol $\relsymbol\in \sigma$ and every $j \in [\uniform]$, there exists $\mathbf{b} \in \alt^{[\uniform] \setminus \{j\}}$ such that %
    $\{(b_1,\ldots,b_{j-1},a,b_{j+1},\ldots,b_{t}) \;|\;a\in \alt\} \subseteq \relsymbol^\relt$. %
\end{definition}

In other words, if $\csp{\relt}$ is lax with input $\rela$%
, then for every relation $\relsymbol^\rela$, every tuple $\mathbf{\verta} \in \relsymbol^\rela$, and every element $\verta_j \in \mathbf{\verta}$, there exists a map $ \homomphi \colon \vars{\mathbf{\verta}} \setminus \{\verta_j\} \to \alt $ such that any extension of $\homomphi$ to $\vars{\mathbf{\verta}}$ maps $\mathbf{\verta}$ to some $\mathbf{\verti} \in \relsymbol^\relt$. Properties similar to but distinct from 
the lax condition, such as \emph{robustness}~\cite{ABRW04Pseudorandom} and \emph{respectful boundary expansion}~\cite{MN24Generalized}, have also been studied in proof complexity.

To introduce the null-constraining property, originally defined by Chan and Molloy~\cite{CM13Dichotomy}, we need a few definitions. Given a $\uniform$-uniform $\sigma$-structure $\relt$, a \emph{simple path instance of $\csp{\relt}$} (of length $\ell$) is an instance $\rela$ of $\csp{\relt}$ such that the hypergraph of $\rela$ is a simple path (of length $\ell$). A simple path instance $\rela$ of $\csp{\relt}$ with endpoints $(\verta_1, \verta_2) \in \ala^2$ \emph{permits} a pair $(\verti_1, \verti_2) \in \alt^2$ if there exists a homomorphism $\homomphi \colon \rela \to \relt$ such that $\homomphi(\verta_1) = \verti_1$ and $\homomphi(\verta_2) = \verti_2$.

\begin{definition}[Null-constraining; \cite{CN25Hierarchies}, Definition~4.2] 
    For $\ell \geq 1$ and a $\uniform$\nobreakdash-uniform $\sigma$\nobreakdash-structure $\relt$, we say that $\csp{\relt}$ is \emph{$\ell$-null-constraining} if for every pair $(\verti_1, \verti_2) \in \alt^2$, it holds that every simple path instance of $\csp{\relt}$ of length at least $\ell$ permits $(\verti_1, \verti_2)$. We say that $\csp{\relt}$ is \emph{null-constraining} if it is $\ell$-null-constraining for some $\ell$. A structure $\relt$ is null-constraining if $\csp{\relt}$ is. %
    
\end{definition}

Note that graph coloring with at least $3$ colors is $2$-null-constraining %
 but not lax. 
An example of a non-trivial lax and null-constraining CSP is %
$\uniform$-uniform hypergraph coloring for $\uniform \geq 3$ and at least $2$ colors. For more examples, see \cite[Section~4.3]{CN25Hierarchies}. 

A \emph{$\uniform$-uniform random hypergraph $H = (A, \hypedgeset)$} with $n$ vertices and $m$ edges is one chosen uniformly at random from the set of $\uniform$-uniform $n$-vertex hypergraphs with $m$ edges.
Given a $\uniform$-uniform $\sigma$\nobreakdash-structure $\relt$ that is lax and null-constraining, the fooling instances for $\csp{\relt}$ are $\sigma$-structures $\rela$ sampled according to the following distribution, which we denote by $\cspdist{n}{m}{\relt}$.  First, pick a $\uniform$-uniform random $n$-vertex hypergraph with $m$ edges. Then, for each $e \in \hypedgeset$, pick a permutation $\pi_e = (\vertv_1, \ldots, \vertv_\uniform)$ of the vertices in $e$ and a relation symbol $\relsymbol_e \in \sigma$, both uniformly at random. Finally, let $\rela$ be the $\sigma$-structure over $\ala$ where, for each $\relsymbol \in \sigma$, we define $\relsymbol^\rela = \{\pi_e : \relsymbol_e = \relsymbol\}$. %

For a $\sigma$-structure $\relt$, we say that $\csp{\relt}$ is \emph{trivially satisfiable} if there exists some $i \in \alt$ such that, for any instance $\rela$ of $\csp{\relt}$, the map $\ala \to \alt$ that sends all elements in $\ala$ to $i$ is a homomorphism $\rela \to \relt$.\footnote{Such CSPs are called \emph{reflexive} in \cite{AD22Width}.} A random instance of such a $\mathrm{CSP}$ is clearly always satisfiable, so to obtain a fooling instance we must assume that $\csp{\relt}$ is not trivially satisfiable. This assumption turns out to also be sufficient to guarantee that randomly sampled instances with linearly many constraints in the number of variables are unsatisfiable with high probability. %
\begin{lemma}[\cite{AD22Width}, Lemma~3]
    \label{lem:trivially-satisfiable}
If $\relt$ is a $\uniform$-uniform relational structure such that $\csp{\relt}$ is not trivially satisfiable, then there exists $\Delta > 0$ depending only on $\relt$ such that, %
with probability $1-o_n(1)$, a structure $\rela \sim \cspdist{n}{\Delta n}{\relt}$ is not homomorphic to $\relt$. 
\end{lemma}
In \cite{AD22Width}, the proof is for random instances whose hypergraphs are sampled from the \emph{Erd\H{o}s-R\'{e}nyi} random model. By, e.g.,~\cite[Lemma 25]{MS07RandomConstraint}, the same result holds for our random hypergraph model.

We are now ready to state \cref{th:cohomological-lax-nullcons-lbs}, which is the main result of this section. %
\begin{theorem}
    \label{th:cohomological-lax-nullcons-lbs}
    Let $\uniform \geq 2$, and let $\relt$ be a $\uniform$-uniform relational structure that is lax and $\ell$-null-constraining, and such that $\csp{\relt}$ is not trivially satisfiable. There exists $\Delta > 0$ depending only on $\relt$, and parameters $\zeta = \zeta(\uniform, \ell, \Delta)> 0$ and $\delta = \delta (\uniform, \ell, \Delta) > 0$ such that an instance $\rela \sim \cspdist{n}{\Delta n}{\relt}$ of $\csp{\relt}$ is unsatisfiable with high probability as $n \to \infty$, but with probability $\delta - o_n(1)$ is accepted by the cohomological $\kcons$\nobreakdash-consistency algorithm for all $\kcons \leq \zeta n$. Consequently, the cohomological $\kcons$\nobreakdash-consistency algorithm does not solve $\csp{\relt}$ for any $\kcons \leq \zeta n$. 
\end{theorem}

The proof of \cref{th:cohomological-lax-nullcons-lbs} goes through \cref{lem:reducibility-gives-reduction-operator}, and the proof of the reducibility condition (\cref{lem:reducibility-nullcons}) applies to any field, not only those of characteristic zero. Hence, we also obtain the same lower bound of $\zeta n$ as in \cref{th:cohomological-lax-nullcons-lbs} %
for the polynomial calculus degree required to refute $\polyhomo{\rela}{\relt}$ over any field under the same assumptions on $\rela$ and $\relt$. 
\begin{corollary}
    \label{cor:cohomological-lax-nullcons-lbs-pc}
    For $\rela, \relt, \zeta, \delta$ as in \cref{th:cohomological-lax-nullcons-lbs}, it holds with probability $\delta - o_n(1)$ that polynomial calculus requires degree at least $\zeta n$ to refute $\polyhomo{\rela}{\relt}$ over any field.
\end{corollary}
To the best of our knowledge, \cref{cor:cohomological-lax-nullcons-lbs-pc} is new and in particular does not follow from \cite[Theorem~1.3]{CN25Hierarchies}.

The rest of this section is devoted to the proof of \cref{th:cohomological-lax-nullcons-lbs}. To begin with, we recall several technical properties of hypergraphs from \cite{CN25Hierarchies}. Just as in \cref{sec:approximate-graph-coloring}, our results rely on a notion of \emph{sparsity}.

\begin{definition}[Sparsity]
    \label{def:sparse-hypergraph}
    A $\uniform$-uniform hyperedge set $\hypedgeset$ over a set $V$ is \emph{$\epsilon$-sparse} if $\setsize{\hypedgeset} \leq \frac{1+\epsilon}{\uniform -1}\setsize{V(\hypedgeset)}$ and is \emph{$(\spsize, \epsilon)$-sparse} if all $\spsize$-small subsets of $\hypedgeset$ are $\epsilon$-sparse. A hypergraph is $(\spsize, \epsilon)$\nobreakdash-sparse if its hyperedge set is.
\end{definition}

With constant probability as $n \to \infty$, random hypergraphs are sparse and have large girth, as stated in the next lemma. Both properties are standard.  %
For the claim regarding sparsity, see, e.g., \cite[Lemma~10]{MS07RandomConstraint}. For the lower bound on girth see, e.g., \cite[Theorem~3.19]{JLR00RandomGraphs}. The calculations there are for graphs, but generalize readily.  %

\begin{lemma}[Properties of random hypergraphs]
    \label{lem:random-hypergraph-properties}
    Let $H$ be a random $\uniform$-uniform hypergraph with $n$ vertices and $\Delta n$ edges. For any $\epsilon > 0$, there exists $\mu = \mu(\Delta,\uniform,\epsilon)$ such that $H$ is $(\mu n, \epsilon)$-sparse with high probability as $n \to \infty$. Moreover, for all $\ell \in \N^+$ there exists $\delta = \delta(\Delta, \ell, \uniform) > 0$ such that, with probability $\delta - o_n(1)$, the hypergraph $H$ has no Berge cycle of length at most $\ell$. 
\end{lemma}

To define closure we need the following notions, which are rephrasings of \cite[Definition~5.20]{CN25Hierarchies} in light of \cite[Remark~5.16]{CN25Hierarchies}.

\begin{definition}[$(U,\nullcons)$-boundary, $(U,\nullcons)$-bad]
    Given $\ell \in \N^+$, a 
    hypergraph $\hypgraph = (V, \hypedgeset)$, %
    subset $U \subseteq V$, and %
    hyperedge subset $\hypedgesubset \subseteq \hypedgeset$, let $B_1(U, \hypedgesubset)$ be the set of edges in $\hypedgesubset$ that contain at most one vertex $u$ such that $u \in U$ or $\deg_\hypedgesubset(u) > 1$, and let $B_2^\ell(U, \hypedgesubset)$ be the set of all pendant paths (each viewed as a set of hyperedges) of length $\ell$ in $\hypedgesubset \setminus \hypedgeset[U]$ whose non-endpoints are all outside $U$. 
    The \emph{$(U,\nullcons)$-boundary} of $\hypedgesubset$ %
     is $B_1(U, \hypedgesubset)\union B_2^\ell(U, \hypedgesubset)$ and is denoted $B^\ell(U, \hypedgesubset)$.
    A set $\hypedgesubset \subseteq \hypedgeset$ such that $B^\ell(U, \hypedgesubset) = \varnothing$ is called \emph{$(U,\nullcons)$-bad}.\footnote{Chan and Ng call such sets \emph{closed}~\cite[Remark 5.16]{CN25Hierarchies}, which unfortunately clashes with our usage of that word.}
\end{definition}

\begin{definition}[Expanding; \cite{CN25Hierarchies}, Definition~5.27]
    \label{def:expanding-hypergraph}
    A $\uniform$-uniform hyperedge set $\hypedgeset$ over a set $V$ is \emph{$(\ell, \spsize, \gamma)$\nobreakdash-expanding} if every $\spsize$-small $\hypedgesubset \subseteq \hypedgeset$ satisfies $\lvert B^\ell(\varnothing, \hypedgesubset)\rvert  \geq \gamma \setsize{\hypedgeset}$. A hypergraph is $(\ell, \spsize, \gamma)$\nobreakdash-expanding if its hyperedge set is.
\end{definition}

If a hypergraph $H$ is sparse and has no short Berge cycles, as is the case for a random hypergraph by \cref{lem:random-hypergraph-properties}, then $H$ is expanding by the next lemma. 

\begin{lemma}[\cite{CN25Hierarchies}, Lemma~5.28]
    \label{lem:sparse-girth-expansion}
    Let $\ell \geq 2$ and let $H=(V, \hypedgeset)$ be a $\uniform$-uniform hypergraph. There exists $\epsilon = \epsilon(\uniform, \ell)$ such that the following holds. If $H=(V, \hypedgeset)$ is $(\spsize, \epsilon)$-sparse and ${\mathrm{girth}(\hypedgeset) > \ell}$, then $H$ is $(\ell, \spsize, \gamma)$\nobreakdash-expanding for $\gamma = \frac{\uniform-1}{72\ell^2\uniform^3(1+\epsilon)}$, provided that $\gamma \leq 1$.   
\end{lemma}

We collect some properties of bad sets in the next lemma. For their (short) proofs, see \cite[Lemmas 5.3, 5.17, and 5.21]{CN25Hierarchies}.

\begin{lemma}
    \label{lem:bad-sets-properties}
    For any hypergraph $\hypgraph = (V, \hypedgeset)$, %
     vertex sets $U\subseteq W \subseteq V$, and $\ell \in \N^+$, the following properties of $(U, \ell)$-bad sets hold. 
    \begin{enumerate}
        \item Any union of $(U, \ell)$-bad hyperedge sets is $(U, \ell)$-bad.
        \item If a hyperedge set $\hypedgesubset\subseteq \hypedgeset$ is $(U, \ell)$-bad, then $\hypedgesubset$ is $(W, \ell)$-bad. 
        \item Chaining property: Given hyperedge sets $\hypedgesubset, \hypedgesubset' \subseteq \hypedgeset$, the set $\hypedgesubset'$ is $(U \union V(\hypedgesubset), \ell)$-bad if and only if $\hypedgesubset' \union \hypedgesubset$ is $(U, \ell)$-bad.
    \end{enumerate}
\end{lemma}

To define closure, Chan and Ng proceed differently from in \cref{sec:approximate-graph-coloring}: building on the work of \cite{KMOW2017Sumofsquares}, given a hypergraph %
$H= (V, \hypedgeset)$ they define the \emph{$s$-local closure}\footnote{This construction is called the \emph{BW closure} in \cite{CN25Hierarchies}.} of $U$ as
\begin{equation}
    \label{eq:local-closure}
\clnullcons{U}{s} = U \union \bigcup_{
    \substack{ 
        \hypedgesubset \subseteq \hypedgeset :\ \setsize{\hypedgesubset}\leq s \\
        \text{$\hypedgesubset$ is $(U, \ell)$-bad}
        } 
        } V(\hypedgesubset). 
\end{equation}
Readers familiar with proof complexity may note the similarity between \eqref{eq:local-closure} and the notion of \emph{support} \cite{AR03LowerBounds,Filmus14AlekhnovichRazborovTCS,MN24Generalized} defined using \emph{unique neighbor expansion}.

By definition, the $s$-local closure is unique, contains $U$, and satisfies monotonicity by \cref{lem:bad-sets-properties}. The next two lemmas show that $\clnullcons{U}{s}$ is small when $U$ is small, and that taking local closure is an idempotent operation.

\begin{lemma}[\cite{CN25Hierarchies}, Lemma~5.33]
    \label{lem:local-closure-small}
    Let $H = (V, \hypedgeset)$ be a hypergraph, let $U \subseteq V$, and let $\hypedgesubset = \hypedgesubset_1 \union \hypedgesubset_2 \subseteq E$ be the union of two $s$-small $(U, \nullcons)$-bad hyperedge sets $\hypedgesubset_1$ and $ \hypedgesubset_2$. If $\hypedgeset$ is $(\ell, 2s, \gamma)$-expanding and $\setsize{U} \leq r\gamma/\ell$ for some $r \geq 0$, then $\hypedgesubset$ is $r$-small.
\end{lemma}

\begin{corollary}[Closure size lemma, cf. \cite{CN25Hierarchies}, Lemma~6.11]
    \label{cor:closure-size}
    If a hypergraph $H=(V, \hypedgeset)$ is $\uniform$-uniform and $(\ell, 2s, \gamma)$-expanding, then for all sets $U \subseteq V$ such that $\setsize{U} \leq r \gamma / \ell$ it holds that 
    $\setsize{\clnullcons{U}{s}} \leq \setsize{U}(1 + \uniform \ell/\gamma)$.
\end{corollary}
\begin{proof}
    Since $\clnullcons{U}{s} = U \union \bigcup\{ V(\hypedgesubset): \hypedgesubset \subseteq \hypedgeset , \text{$\setsize{\hypedgesubset}\leq s$, $\hypedgesubset$ $(U, \ell)$-bad in $\hypedgeset$}\}$ and the union of two $(U, \ell)$\nobreakdash-bad sets is again $(U, \ell)$-bad by \cref{lem:bad-sets-properties}, the claim follows readily from induction on the number of $(U, \ell)$-bad sets coupled with  \cref{lem:local-closure-small}.
\end{proof}

\begin{lemma}[Small-set idempotence of $\clnullcons{\cdot}{s}$] 
    \label{lem:small-set-idempotence}
    If a hypergraph $H=(V, \hypedgeset)$ is $\uniform$-uniform and $(\ell, 2s, \gamma)$\nobreakdash-expanding, then for all sets $U \subseteq V$ such that $\setsize{\clnullcons{U}{s}} \leq s \gamma/\ell$, it holds that 
    $\clnullcons{\clnullcons{U}{s}}{s} = \clnullcons{U}{s}$. 
\end{lemma}
\begin{proof}
    Clearly $\clnullcons{U}{s} \subseteq \clnullcons{\clnullcons{U}{s}}{s}$, so our goal is to establish the converse inclusion. Denote $\clnullcons{U}{s}$ by $T$ and let $\hypedgebadset$ denote the union of $s$-small $(U, \ell)$-bad sets. Note that $\hypedgebadset$ is $(U, \ell)$\nobreakdash-bad by \cref{lem:bad-sets-properties} and hence $(T, \ell)$-bad, again by \cref{lem:bad-sets-properties}. Moreover, $\hypedgebadset$ is $s$-small by \cref{cor:closure-size} and \cref{lem:local-closure-small}. 
    
    Consider any $s$-small $(T, \ell)$-bad set $\hypedgesubset \subseteq \hypedgeset$. Since $T$ is $s \gamma/\ell$-small, the set $\hypedgesubset \union \hypedgebadset$ is $s$-small by \cref{lem:local-closure-small} and is $(U, \ell)$-bad by the chaining property in \cref{lem:bad-sets-properties}. Thus, $\hypedgesubset \subseteq \hypedgesubset \union \hypedgebadset\subseteq \hypedgebadset$ by the definition of $\hypedgebadset$. It follows that the union $\hypedgebadsetalt$ of $s$-small $(T, \ell)$-bad sets is contained in $\hypedgebadset$. Since $T = U \union V(\hypedgebadset)$, we thus have that $\clnullcons{\clnullcons{U}{s}}{s} = (T \union V(\hypedgebadsetalt)) \subseteq (U \union V(\hypedgebadset)) = \clnullcons{U}{s}$ as desired. \qedhere        
    
\end{proof}

The final ingredient needed to establish the reducibility condition is \emph{insularity}~\cite[Section~6]{CN25Hierarchies}. Intuitively, a vertex set $U$ in a hypergraph $H = (V, \hypedgeset)$ is insular with respect to an edge set $\hypedgesubset$ if $\hypedgesubset$ has no additional $(U, \ell)$-bad edge subsets beyond those already present in $\hypedgeset[U]$.

\begin{definition}[Insular, cf. \cite{CN25Hierarchies}, Definition~6.1 and Section~7.1]
    \label{def:insular}
Let $H=(V, \hypedgeset)$ be a hypergraph, %
$U\subseteq V$, and $\hypedgesubset \subseteq \hypedgeset$. 
We say that $U$ is \emph{$\ell$-insular in $\hypedgesubset$} if all $(U, \ell)$-bad hyperedge sets in $\hypedgesubset$ are contained in $\hypedgeset[U]$.
\end{definition}

\begin{observation}
    \label{obs:local-closure-locally-closed}
    Let $H=(V, \hypedgeset)$ be a hypergraph, %
     $U \subseteq V$, 
     and $\hypedgesubset \subseteq \hypedgeset$. If $\clnullcons{U}{s} = U$  %
     and
    $\hypedgesubset$ is $s$-small, then $U$ is $\ell$-insular in $\hypedgesubset$.  %
\end{observation}

\begin{proof}
    
    If $\hypedgebadset \subseteq \hypedgesubset$ is $(U, \ell)$-bad, then since 
    $\hypedgesubset$ is $s$-small, we have that $\hypedgebadset  \subseteq \hypedgeset[\clnullcons{U}{s}]$ by the definition of $\clnullcons{\cdot}{s}$. Since $\clnullcons{U}{s} =U$, it follows that $\hypedgebadset \subseteq \hypedgeset[U]$. \qedhere
\end{proof}

The next lemma is used to establish the reducibility condition, and states that if a set $U$ is $\ell$\nobreakdash-insular in $\hypedgesubset$, then the neighborhood of $U$ in $\hypedgesubset$ is $\Omega(\ell)$-insular in $\hypedgesubset$ given some technical assumptions on $\hypedgesubset$. \cref{lem:neighborhood-no-bad-sets} is proved in \cite[Section~7.1]{CN25Hierarchies}. %

\begin{lemma}
    \label{lem:neighborhood-no-bad-sets}
    Let $\ell\in \N^+$, let $H=(V, \hypedgeset)$ be a $\uniform$-uniform hypergraph, and let $U \subseteq V$ and $\hypedgesubset \subseteq \hypedgeset$. Denote $U \union \nbhstd{\hypedgesubset}{U}$ by $W$. %
    There exists $\epsilon=\epsilon(\uniform, \ell) > 0$ such that the following holds. If $\hypedgesubset$ is $\epsilon$-sparse and $\mathrm{girth}(\hypedgesubset) > 3\ell$, then if $U$ is $(3\ell+2)$-insular in $\hypedgesubset$, it holds that $W$ is $\ell$-insular in $\hypedgesubset$. 
\end{lemma}

The following lemma is \cite[Lemma~6.19]{CN25Hierarchies}, which is implicit in \cite[Lemma~5.8]{CM13Dichotomy}. It is used both for the satisfiability and reducibility conditions.

\begin{lemma}[Extension lemma]
    \label{lem:extension-lemma}
     Let $\relt$ be a uniform, $\ell$-null-constraining relational structure, let $\rela$ be an instance of $\csp{\relt}$, and let $H=(\ala, \hypedgeset)$ be the hypergraph of $\rela$. For any set $U \subseteq \ala$, hyperedge set $\hypedgesubset \subseteq \hypedgeset$, and hypergraph $h = (V(h), \hypedgeset(h)) \in B^\ell(U, \hypedgesubset)$, any partial homomorphism $\rho\colon \rela\to \relt$ with domain $U \union V(\hypedgesubset \setminus \hypedgeset(h))$ %
       can be extended to a partial homomorphism $\rho'\colon \rela\to \relt$ with domain $U \union V(\hypedgesubset)$.
\end{lemma}

\begin{proof}
If $h \in B_1(U, \hypedgesubset)$, then $h$ is a hyperedge that intersects $U \union V(\hypedgesubset \setminus h)$ in at most one vertex, say $u$. Thinking of $u$ as the endpoint of a simple path instance of length $\ell$, the null-constraining property yields a partial homomorphism on $h$. %

If $h\in B_2^\ell(U, \hypedgesubset) \setminus B_1(U, \hypedgesubset)$%
, then $h$ is a pendant path of length $\ell$ in $\hypedgesubset$ whose non-endpoints are outside $U \union V(\hypedgesubset \setminus \hypedgeset(h))$. By the null-constraining property, we obtain a partial homomorphism on $h$ that agrees with $\rho$ on $(U \union V(\hypedgesubset \setminus \hypedgeset(h)))\intersection V(h)$. 

In both cases, combining $\rho$ with the map on $h$ yields the desired $\rho'$. 
\end{proof}

\begin{corollary}
    \label{cor:full-extension-lemma}
    Let $\relt$ be a uniform, $\ell$-null-constraining relational structure, let $\rela$ be an instance of $\csp{\relt}$, and let $H=(A, \hypedgeset)$ be the hypergraph of $\rela$. Let 
    $U \subseteq \ala$ and $\hypedgesubset \subseteq \hypedgeset$. If $U$ is $\ell$-insular in $\hypedgesubset$, then every partial homomorphism on $U$ can be extended to $U \union V(\hypedgesubset)$. %
\end{corollary}

\begin{proof}
    We proceed by induction on the size of $\hypedgesubset$. For the base case $\setsize{\hypedgesubset} = 0$ there is nothing to prove. For the induction step, suppose that the claim holds for all hyperedge sets of size at most $s-1$, and consider a set $\hypedgesubset$ of size $s$. Since $U$ is $\ell$-insular in $\hypedgesubset$, it holds that $B^\ell(U, \hypedgesubset)$ is non-empty, and hence contains some hypergraph $h = (V_h, \hypedgeset_h)$. Let $\hypedgesubset' = \hypedgesubset \setminus \hypedgeset_h$. Then $U$ is $\ell$-insular in $\hypedgesubset'$. By the induction hypothesis, any partial homomorphism on $U$ can be extended to $U \union V(\hypedgesubset')$, and by \cref{lem:extension-lemma} also to $U \union V(\hypedgesubset)$. \qedhere
    
\end{proof}

Now we are finally in a position to establish the reducibility condition. The proof uses similar ideas as that of \cite[Theorem~7.13]{CN25Hierarchies}.

\begin{lemma}[Reducibility condition]
    \label{lem:reducibility-nullcons}
     Let $\uniform \geq 2$, $\ell\geq 1$, %
     and let $\relt$ be a $\uniform$-uniform, lax, and $\ell$-null-constraining relational structure. 
     Let $\rela$ be an instance of $\csp{\relt}$ and let $H=(A, \hypedgeset)$ be the hypergraph of $\rela$. 
     Suppose that $\operatorname{girth}(\hypedgeset)> 3\ell + 2$, and let $\epsilon = \epsilon(\uniform, \ell)$ be maximal such that \cref{lem:neighborhood-no-bad-sets} holds, let $s$ be maximal such that $H$ is $(2s, \epsilon)$-sparse, let $\ell' = 3\ell+2$, and let $\gamma = \frac{\uniform-1}{72\ell'^2\uniform^3(1+\epsilon)}$.  

    Let $\F$ be any field and consider any admissible
  order on the monomials in $\F[\mathbf{x}_{\ala, \alt}]$. For all monomials $m \in \F[\mathbf{x}_{\ala, \alt}]$ of degree at most $\degstd \leq s\gamma /(\ell'(1 + \uniform \ell'/\gamma))$, it holds for all $m'$ such that $\clnull{m'}{\ell'}{s} \subseteq \clnull{m}{\ell'}{s}$ that $m'$ is reducible modulo $\ideal{\clnull{m'}{\ell'}{s}}$ if and only if $m'$ is reducible modulo $\ideal{\clnull{m}{\ell'}{s}}$. 
\end{lemma}

    \begin{proof}
    For brevity, denote $\clnull{m'}{\ell'}{s}$ by $\setU$ and $\clnull{m}{\ell'}{s}$ by $\setW$. Since $\setU \subseteq \setW$, if $m'$ is reducible modulo $\ideal{\setU}$ then it is reducible modulo $\ideal{\setW}$. 
    For the converse direction, suppose that $m'$ is reducible modulo $\ideal{\setW}$, which is to say that there exists a polynomial $p \in \ideal{\setW}$ such that $m'$ is the monomial in the leading term of $p$. %
    The idea is to apply a variable restriction $\rho$ such that $\restrict{p}{\rho}$ is a polynomial in the ideal $\ideal{\setU}$ with $m'$ as the leading monomial. 
    To do so, we first establish some technical properties of $\setU$ and $\setW$. 

    By \cref{lem:sparse-girth-expansion}, the hypergraph $H$ is $(\ell', 2s, \gamma)$-expanding. Therefore, given $\setsize{\verts{V}(m)}\leq \degstd$ and $\setU\subseteq \setW$, we have by \cref{cor:closure-size} that $\setsize{U} \leq s\gamma/\ell'$, and hence $\clnull{\setU}{\ell'}{s} = \setU$ by \cref{lem:small-set-idempotence}. 
    Moreover, the set $\hypedgeset[\setW]$ is $s$-small by sparsity and \cref{lem:local-closure-small}, so by \cref{obs:local-closure-locally-closed} we have that $\setU$ is $\ell'$-insular in $\hypedgeset[\setW]$, and hence $\setU \union \nbhstd{\hypedgeset[\setW]}{\setU}$ is $\ell$-insular in $\hypedgeset[\setW]$ by \cref{lem:neighborhood-no-bad-sets}. 

    Consider the set of edges in $\hypedgeset[W]$ that nontrivially intersect $U$, denoted $\mathcal F=\{e\in \hypedgeset[\setW]\setminus\hypedgeset[\setU]: e\cap \setU\neq\varnothing\}$. %
    Each edge $e \in \mathcal F$ intersects $\setU$ in exactly one vertex, since otherwise $\{e\}$ is $(\setU, \ell')$-bad. 
    Moreover, there is no 2- or 3-path, nor 2- or 3-cycle, in $\hypedgeset[\setW] \setminus \hypedgeset[\setU]$ whose endpoints are in $\setU$, as such a set of edges would be $(\setU, \ell')$-bad given $\ell'>3$. 
    It follows that the edges in $\mathcal F$ are pairwise vertex disjoint outside of $\setU$ and that each edge in $\hypedgeset[\setW]\setminus \hypedgeset[\setU]$ intersects at most one edge in $\mathcal F$. 
    Moreover, since $\hypedgeset$ is $\uniform$-uniform, the hyperedge set $\hypedgeset[\nbhstd{\hypedgeset[\setW]}{\setU}]$ must be empty, as otherwise there would be an edge in $\hypedgeset[\setW] \setminus \hypedgeset[\setU]$ intersecting at least two edges in $\mathcal F$. 

    We are ready to define $\rho$. For each edge $e \in \hypedgesubset$, 
    the lax property yields a partial homomorphism on $V(e) \setminus \setU$ that satisfies $e$. By the above, the union of these homomorphisms is a well-defined partial homomorphism $\tau$ on all of $\mathcal F$, whose domain is $\nbhstd{\hypedgeset[\setW]}{\setU}$. 
    We extend $\tau$ to all of $\hypedgeset[W]\setminus (\mathcal F \union \hypedgeset[\setU])$ as follows.
    Since any $(\nbhstd{\hypedgeset[\setW]}{\setU},\ell)$-bad subset of $\hypedgeset[W]\setminus(\mathcal F \union \hypedgeset[\setU])$ is also $(\setU \union \nbhstd{\hypedgeset[\setW]}{\setU},\ell)$-bad (\cref{lem:bad-sets-properties}), by the $\ell$-insularity of $\setU \union \nbhstd{\hypedgeset[\setW]}{\setU}$ in $\hypedgeset[\setW]$ (recorded two paragraphs above), the vertices of edges in this bad set are contained in $\nbhstd{\hypedgeset[\setW]}{\setU}$, but such edges cannot exist as noted at the end of the previous paragraph. %
    This means $\nbhstd{\hypedgeset[\setW]}{\setU}$ is $\ell$-insular in $\hypedgeset[W]\setminus(\mathcal F \union \hypedgeset[\setU])$. 
    Therefore, we can extend $\tau$ to vertices in $\hypedgeset[W]\setminus \hypedgeset[U]$ satisfying all of $\hypedgeset[\setW] \setminus \hypedgeset[\setU]$ by \cref{cor:full-extension-lemma}; this assignment—or more precisely, its Boolean encoding as a partial assignment to the variables in $\mathbf{x}_{\ala, \setT}$—is our $\rho$. 
    Note that the domain of $\rho$ is disjoint from $\mathbf{x}_{\setU,\setT}$.

    It remains to show two facts: (1) $\restrict{p}{\rho} \in \ideal{\setU}$, and (2) $m'$ is its leading monomial. 
    Fact (1) holds since $p \in \ideal{W}$ and $\rho$ leaves the generating polynomials of $\ideal{U}$ untouched while setting the remaining generating polynomials of $\ideal{W}$ to $0$. 
    To see fact (2), we notice that $\restrict{m'}{\rho}=m'$ (since $\verts{\ala}(m') \subseteq U$ and $\rho$ assigns no variable over $U$), and for the remaining monomials in $p$, applying variable restrictions never increases their order in an admissible order. 
    \end{proof}

We now have all components needed to prove \cref{th:cohomological-lax-nullcons-lbs}.

\begin{proof}[Proof of \cref{th:cohomological-lax-nullcons-lbs}]
     The existence of $\Delta> 0$ depending only on $\relt$ such that an instance $\rela$ sampled from $\cspdist{n}{\Delta n} {\relt}$ is unsatisfiable with high probability as $n \to \infty$ follows from \cref{lem:trivially-satisfiable}.
     
     Fix $\rela \sim \cspdist{n}{\Delta n} {\relt}$, a field $\F$ of characteristic $0$, and any lexicographic order on the monomials of $\F[\mathbf{x}_{\ala, \alt}]$. Let $H = (\ala, \hypedgeset)$ be the hypergraph of $\rela$. Let $\epsilon = \epsilon(\uniform, \ell)$ be maximal such that \cref{lem:neighborhood-no-bad-sets} holds whenever $\operatorname{girth}{\hypedgeset} > 3\ell + 2$, let $s$ be maximal such that $H$ is $(2s, \epsilon)$-sparse, let $\ell' = 3\ell+2$, and let $\gamma = \frac{\uniform-1}{72\ell'^2\uniform^3(1+\epsilon)}$. 
    
    By \cref{lem:random-hypergraph-properties}, there exists $\delta = \delta(\Delta, \uniform, \ell) > 0$ such that, with probability $\delta - o_n(1)$, it holds that $\operatorname{girth}(H) > 3\ell + 2$ and that there exists $\mu = \mu(\Delta,\uniform,\epsilon) > 0$ such that $s \geq \mu n$. %
    Then $H$ is $(\ell', 2s, \gamma)$-expanding by \cref{lem:sparse-girth-expansion}. Assuming these properties hold, we show that $\clnull{\cdot}{\ell'}{s}$ is a closure operator satisfying the satisfiability and reducibility conditions for all monomials $m \in \F[\mathbf{x}_{\ala, \alt}]$ of degree at most $\degstd \leq s\gamma /(\ell'(1 + \uniform \ell'/\gamma))$.
    
    Let $m$ be a monomial of degree at most $\degstd$. For all $\degstd$-small sets $U\subseteq \ala$, we have by \cref{lem:closure-size-lemma} that $\setsize{\clnull{U}{\ell'}{s}} \leq s\gamma/\ell'$. Hence, by \cref{lem:small-set-idempotence}, the operator $\clnull{\cdot}{\ell'}{s}$ is idempotent on all $\degstd$-small vertex sets, and is therefore a size-$\degstd$ closure operator.

    The reducibility condition for $\clnull{\cdot}{\ell'}{s}$ on monomials of degree at most $\degstd$ is precisely \cref{lem:reducibility-nullcons}, so it remains to establish the satisfiability condition. Let $m$ be a monomial of degree $\degstd$. Since $\setsize{\verts{V}(m)} \leq \degstd$, we have by \cref{cor:closure-size} that $\setsize{\clnull{m}{\ell'}{s}} \leq s\gamma/\ell'$. Moreover, the set $\hypedgeset[\clnull{m}{\ell'}{s}]$ is $s$-small by sparsity and \cref{lem:local-closure-small}. Since $H$ is $(\ell', 2s, \gamma)$-expanding, it follows by \cref{cor:full-extension-lemma} (instantiated with $U = \varnothing$, which is $\ell'$-insular in $\hypedgeset[\clnull{m}{\ell'}{s}]$ by expansion)
    that there exists a satisfying assignment to $\hypedgeset[\clnull{m}{\ell'}{s}]$. 
    
    With the satisfiability and reducibility conditions established for all monomials of degree at most~$\degstd$, it follows from \cref{cor:closure-size} and \cref{lem:cohomological-k-consistency-lower-bounds} that the cohomological $\kcons$\nobreakdash-consistency algorithm for $\csp{\relt}$ accepts $\rela$ for all $\kcons\leq \degstd / (1 + \uniform \ell'/\gamma) \leq s\gamma /(\ell'(1 + \uniform \ell'/\gamma)^2) =\zeta n$, for some $\zeta = \zeta(\Delta, \uniform, \ell)$. The proof is complete.
\end{proof}

\section{Concluding Remarks}
\label{sec:concluding-remarks}
We show that pseudo-reduction operators from proof complexity that are constructed via the Alekhnovich-Razborov method can be used to prove level lower bounds for (P)CSP hierarchies. 
Using this %
connection, we prove optimal level lower bounds for the cohomological $\kcons$-consistency algorithm for approximate graph coloring, and give a new proof of the %
level lower bounds in~\cite{CN25Hierarchies} for lax and null-constraining CSPs.

Perhaps the most pressing open problem %
is whether the cohomological $\kcons$-consistency algorithm solves all tractable CSPs. 
A closely related question is whether technical connections from proof complexity, such as the one presented in this work, can be useful to prove hierarchy lower bounds for tractable CSPs, and 
 whether they can %
 be used to extend our lower bounds to other hierarchies such as %
 the C(BLP+AIP) hierarchy~\cite{CZ23CLAP,CN25Hierarchies}. %

It is also natural to ask whether our results for approximate graph coloring can be generalized to \emph{approximate graph homomorphism}, which is null-constraining but not lax. Lower bounds for this problem have been established for multiple hierarchies, the strongest of which is SDA~\cite{CZ24Semidefinite}, but bounds with parameters comparable to those in \cref{th:cohomological-agc-lbs} are only known for the $\kcons$\nobreakdash-consistency algorithm~\cite{CZ24Periodic} (see also \cite[Theorem~6.24]{CN25Hierarchies}). %

Unlike our results for lax and null-constraining CSPs, it seems unclear how to prove our results for approximate graph coloring without using the connection to polynomial calculus degree lower bounds. In particular, the definition of closure there depends on the choice of a vertex order, which in turn seems motivated by the reduction operator. It would be interesting to see a proof of \cref{th:cohomological-agc-lbs} that does not use pseudo-reduction operators; such a proof would likely also yield insights applicable to polynomial calculus degree lower bounds.

More technically, the use of the lex game in \cref{sec:ipr-operators-lex-game} limits the applicability of our results to lexicographic orders on the monomials. Can we show similar results for, say, the graded lexicographic order? This would likely require different techniques from those in this paper, or might simply be impossible; see \cref{rmk:require-lex-order}.%

\section*{Acknowledgements}
We would like to thank Siu-On Chan, Nicola Galesi, Tamio-Vesa Nakajima, Jakub Opr\v{s}al, Per Austrin, Kilian Risse, and Benedikt Pago for insightful discussions. 
The initial stages of this work were carried out when the authors attended the 2025 \emph{Complexity as a Kaleidoscope} research school at CIRM, Marseille, France. We thank the organizers for creating such an inspiring environment, and Noel Arteche, David Engstr\"{o}m, Stefan Grosser, and Santiago Guzm\'{a}n Pro for helpful conversations while there. 
We are also grateful to the anonymous \emph{SODA} reviewers for their helpful comments, all of which improved the exposition of the paper. 

JC was partially supported by the \emph{Wallenberg AI and Autonomous Systems Program (WASP)} and would also like to thank Susanna de Rezende for encouraging him to apply for funding to attend %
the CIRM research school. 
YG received funding from the Independent Research Fund Denmark grant 9040-00389B.

\appendix

\section{Proof of \cref{lem:reducibility-gives-reduction-operator}}
\label{app:reducibility-gives-reduction-operator}

In this appendix we prove \cref{lem:reducibility-gives-reduction-operator}, restated below for convenience. As noted in \cref{sec:ar-method-cohomological-consistency}, the argument is standard and hence we make no claim of originality. We follow the presentation of Conneryd, de Rezende, Nordstr\"{o}m, Pang, and Risse~\cite[Lemma~4.6]{CdRNPR25GraphColouring}. 
\begin{reducibility-gives-reduction-operator}[Restated]
   Let $\degstd \in \N$, let $\sigma$ be a signature whose relation symbols all have arity at most $\degstd$, let $\rela$ and~$\relt$ be $\sigma$-structures, and let $\operatorname{cl} \colon 2^\ala \to 2^\ala$ be a size-$\degstd$ closure operator for $\ala$. Finally, let $\prec$ be an admissible ordering of the monomials in $\F[\mathbf{x}_{A, T}]$. Suppose that the following conditions hold for every monomial $m \in \F[\mathbf{x}_{\ala, \alt}]$ of degree at most $\degstd$:
    \begin{description}
        \item \emph{Satisfiability:} It holds that $\rela[\cl{m}] \to \relt$.
        \item \emph{Reducibility:} For every monomial $m'$ such that $\cl{m'} \subseteq \cl{m}$, it holds that $m'$ is reducible modulo $\ideal{\cl{m'}}$ if and only if $m'$ is reducible modulo $\ideal{\cl{m}}$.
    \end{description}
    Then, the operator $\pseudred$ on $\F[\mathbf{x}_{A, T}]$ defined to be the $\F$-linear extension of the map that takes each monomial $m \in \F[\mathbf{x}_{A, T}]$ to $\redop^\prec_{\ideal{\cl{m}}}(m)$ is a degree-$\degstd$ pseudo-reduction operator for $\polyhomo{\rela}{\relt}$. 
\end{reducibility-gives-reduction-operator}

Before commencing the proof of \cref{lem:reducibility-gives-reduction-operator}, we record a basic fact of reduction operators as well as two technical claims that both concern how closure operators interact with reduction modulo ideals.

  \begin{fact}
    \label{fact:reduction-operator}
    Let $I$ and $J$ be ideals in $\F[x_1, \ldots, x_n]$ such that $I \subseteq J$. For all monomials~$m$ and~$m'$, it holds that $R_{J}(m'R_{I}(m)) = R_{J}(m'm)$.
  \end{fact}
  \begin{proof}
    Write $m = q + r$ where $q \in I$ and $r$ is the reduction of $m$ modulo $I$. Then $m'R_{I}(m) = m'r$, and also $m'R_{I}(m) = q' + r'$ where $q'\in J$ and $r'$ is the reduction of $m'R_{I}(m)$ modulo $J$. Observe that $R_{J}(m'R_{I}(m)) = r'$, so the claim follows if we show that $R_{J}(m'm) = r'$. Note that $m'R_{I}(m) = m'(m-q) = q' + r'$, and hence $m'm =  q' + r' - m'q$. Since $I \subseteq J$ it holds that $m'q \in J$. Therefore, we can write $m'm = q' - m'q + r'$ where $q' - m'q \in J$ and $r'$ is irreducible modulo $J$. By uniqueness of this representation, it follows that $R_{J}(m'm) = r'$. 
  \end{proof}

\begin{claim}
    \label{claim:reducibility-of-reductions}
    If the conditions of \cref{lem:reducibility-gives-reduction-operator} are satisfied, then for every monomial $m$ of degree at most $\degstd$ and every monomial $m'$ that appears in $\redop^\prec_{\ideal{\cl{m}}}(m)$, it holds that $\cl{m'} \subseteq \cl{m}$.  
  \end{claim}
  \begin{proof}
    It suffices to show that $\verts{\ala}(m')\subseteq \cl{m}$, since monotonicity and idempotence of $\cl{\mathord{\cdot}}$ then imply that $\cl{m'} \subseteq \cl{\cl{m}} = \cl{m}$. The claim trivially follows if $m'=m$, so suppose this is not the case. Toward a contradiction, suppose that $\verts{\ala}(m') \not \subseteq \cl{m}$. Denote the assignment that sets $x_{\verta, \verti}$ to $0$ for all $\verti \in \relt$ and all $\verta \not \in \cl{m}$ by $\reststd$. Denote the polynomial $\redop^\prec_{\ideal{\cl{m}}}(m)$ by $r$, and recall that $r$ is the unique sum of irreducible terms such that $m = r + q$, for some $q \in \ideal{\cl{m}}$. No variable of $m$ nor of any generator of $\ideal{\cl{m}}$ is assigned by $\reststd$, so $\restrict{m}{\reststd} = m$ and $\restrict{q}{\reststd} \in \ideal{\cl{m}}$. However, we have that $\restrict{m'}{\reststd} = 0$, so $\restrict{r}{\reststd} \neq r$. Moreover, all terms in $\restrict{r}{\reststd}$ are irreducible modulo $\ideal{\cl{m}}$, as otherwise $r$ would contain a reducible term. But this contradicts that the representation $m = r + q$ is unique, and thus we conclude that $\verts{\ala}(m')\subseteq \cl{m}$. The claim follows.
  \end{proof}
The next claim states that if the reducibility condition in \cref{lem:reducibility-gives-reduction-operator} is satisfied, then reducing a monomial $m$ modulo $\ideal{\cl{m}}$ and reducing $m$ modulo a larger ideal $\ideal{X}$ generated by a set $X\supseteq \cl{m}$ results in the same polynomial. Thus, the subinstance $\rela[\cl{m}]$ can be said to be the part of $\rela$ that is ``relevant'' for reducing $m$.
\begin{claim}
    \label{claim:reductions-the-same}
    Let $\rela, \relt, \prec$, and $\cl{\cdot}$ be defined as in \cref{lem:reducibility-gives-reduction-operator}, and let $X\subseteq \ala$. Suppose that for all monomials $m \in \F[\mathbf{x}_{\ala, \alt}]$ such that $\cl{m} \subseteq X$, it holds that $m$ is irreducible modulo $\ideal{X}$ if $m$ is irreducible modulo $\ideal{\cl{m}}$. Then, for all monomials $m'$ such that $\cl{m'}\subseteq X$, it holds that $\redop^{\prec}_{\ideal{\cl{m'}}}(m') = \redop^{\prec}_{\ideal{X}}(m')$.
\end{claim}
\begin{proof}
    Let $m$ be a monomial such that $\cl{m}\subseteq X$, and write $\redop^{\prec}_{\ideal{\cl{m}}}(m) = \sum_j a_j m_j$. Recall that $m_j$ is irreducible modulo $\ideal{\cl{m}}$ for all $j$ by definition. Moreover, we have by \cref{claim:reducibility-of-reductions} that $\cl{m_j}\subseteq \cl{m}$, and each $m_j$ is also irreducible modulo $\ideal{\cl{m_j}}$. Therefore, it follows by assumption that each $m_j$ is irreducible modulo $\ideal{X}$. By definition, the polynomial $\redop^{\prec}_{\ideal{\cl{m}}}(m)$ is the unique sum of irreducible terms modulo $\ideal{\cl{m}}$ such that $m = \redop^{\prec}_{\ideal{\cl{m}}}(m) + q$, for some $q\in \ideal{\cl{m}}$. Note that $\cl{m} \subseteq \ideal{X}$ and hence $q \in \ideal{X}$. Since all terms in $\redop^{\prec}_{\ideal{\cl{m}}}(m)$ are also irreducible modulo $\ideal{X}$, it follows by uniqueness that $\redop^{\prec}_{\ideal{\cl{m}}}(m) = \redop^{\prec}_{\ideal{X}}(m)$. 
\end{proof}

With Claims~\ref{claim:reducibility-of-reductions} and \ref{claim:reductions-the-same} established, we are ready to prove \cref{lem:reducibility-gives-reduction-operator}.

\begin{proof}[Proof of \cref{lem:reducibility-gives-reduction-operator}]
Let $\pseudred$ be the operator in the lemma statement. We need to show that $\pseudred$ satisfies Properties~\ref{item:rop-property-1}\nobreakdash-\ref{item:rop-property-3} of a pseudo-reduction operator, that is, that (1)~$\pseudred(1) = 1$, (2)~$\pseudred(\polyp) = 0$ for every polynomial
   $\polyp\in\polyhomo{\rela}{\relt}$, and (3)~$\pseudred(x_{\verta, \verti}\cdot m) =
    \pseudred\bigl(x_{\verta, \verti} \cdot
    \pseudred(m)\bigr)$ for every monomial~$m\in \F[\mathbf{x}_{A, T}]$ of degree at most
  $\pcdeg-1$ and every variable~$x_{\verta, \verti}$.

  To see that $\pseudred(1) = 1$, note that $\rela[\cl{1}] \to \relt$ by the satisfiability condition, and hence $1 \not \in \ideal{\cl{1}}$ by the Boolean Nullstellensatz (\cref{lem:boolean-ns}). Since $1$ is minimal in $\prec$, it follows that $\pseudred(1) = 1$. 

  To establish Property~\ref{item:rop-property-2}, observe that $\cl{m}$ contains all vertices in $\ala$ mentioned by $m$, which implies that $\pseudred$ maps all axioms \eqref{eq:boolean-encoding2}\nobreakdash--\eqref{eq:boolean-encoding4} to $0$. To see that $\pseudred$ also maps the axioms \eqref{eq:boolean-encoding1} to $0$, further note that $\cl{1} \subseteq \cl{\{\verta\}}$ for all $\verta \in \ala$, and hence
  \begin{subequations}
    \begin{align}
    \pseudred\Bigl(\sum_{\verti\in \alt} x_{\verta, \verti} - 1\Bigr) 
    &= \sum_{\verti\in \alt} \pseudred^{\prec}_{\ideal{\cl{x_{\verta, \verti}}}} (x_{\verta, \verti}) - \redop^{\prec}_{\ideal{\cl{1}}}(1) &&\quad [\text{definition of $\pseudred$}] \\
    &= \sum_{\verti\in \alt} \redop^{\prec}_{\ideal{\cl{x_{\verta, \verti}}}} (x_{\verta, \verti}) - \redop^{\prec}_{\ideal{\cl{x_{\verta, \verti}}}}(1) &&\quad [\text{reducibility and \cref{claim:reductions-the-same}}] \\
    &= 0 \eqcomma &&\quad [\text{definition of $\redop^{\prec}_{\ideal{\cl{x_{\verta, \verti}}}}$}]
    \end{align}
  \end{subequations}
  as desired.   

  Let us now establish Property~\ref{item:rop-property-3}, which is subtler than the first two. Let $m$ be a monomial of degree at most $\degstd-1$, and write the polynomial $\pseudred(m)$ as $\sum_j a_j m_j$ where $a_j \in \F$. We have by definition that
  \begin{equation}
    \pseudred(x_{\verta, \verti}\pseudred(m)) = \sum_j a_j \pseudred(x_{\verta, \verti}  m_j) = \sum_j a_j \redop^\prec_{\ideal{\cl{x_{\verta, \verti}  m_j}}}(x_{\verta, \verti} m_j) \eqperiod
  \end{equation}
  The next step is to show that $\redop^\prec_{\ideal{\cl{x_{\verta, \verti}  m_j}}}(x_{\verta, \verti} m_j) = \redop^\prec_{\ideal{\cl{x_{\verta, \verti}  m}}}(x_{\verta, \verti} m_j)$ for all $j$, for which we need \cref{claim:reducibility-of-reductions}. Indeed, note that
  \begin{subequations}
    \begin{align}
    \cl{x_{\verta, \verti}  m_j} 
        &\subseteq \cl{\{\verta\} \union \cl{m_j}} &&\quad [\text{monotonicity and self-containment}] \\
        &\subseteq \cl{\{\verta\} \union \cl{m}} &&\quad [\text{monotonicity and \cref{claim:reducibility-of-reductions}}] \\
        &\subseteq \cl{\cl{x_{\verta, \verti} m}} &&\quad [\text{monotonicity and self-containment}] \\
         &= \cl{x_{\verta, \verti} m}\eqperiod &&\quad [\text{idempotence}]
    \end{align}
  \end{subequations}
By reducibility and \cref{claim:reductions-the-same}, we obtain that $\redop^\prec_{\ideal{\cl{x_{\verta, \verti}  m_j}}}(x_{\verta, \verti} m_j) = \redop^\prec_{\ideal{\cl{x_{\verta, \verti}  m}}}(x_{\verta, \verti} m_j)$ for all~$j$. Summing up, we have thus far established that 
  \begin{equation}
    \label{eq:subtle-step}
    \pseudred(x_{\verta, \verti}\pseudred(m)) = \sum_j a_j \redop^\prec_{\ideal{\cl{x_{\verta, \verti}  m}}}(x_{\verta, \verti} m_j) \eqperiod
  \end{equation}
  From \eqref{eq:subtle-step}, we can conclude that Property~\ref{item:rop-property-3} holds. Indeed, we have
  \begin{subequations}
    \begin{align}
        \sum_j a_j \redop^\prec_{\ideal{\cl{x_{\verta, \verti}  m}}}(x_{\verta, \verti} m_j) 
        &= \redop^\prec_{\ideal{\cl{x_{\verta, \verti}  m}}} \Bigl(x_{\verta, \verti} \sum_j a_j  m_j\Bigr) &&\quad [\text{linearity}] \\
        &= \redop^\prec_{\ideal{\cl{x_{\verta, \verti}  m}}} (x_{\verta, \verti} \redop^\prec_{\ideal{\cl{m}}}(m)) &&\quad [\text{definition of $\textstyle{\sum}_j a_j m_j$}] \\
        &= \redop^\prec_{\ideal{\cl{x_{\verta, \verti}  m}}} (x_{\verta, \verti} m) &&\quad [\text{monotonicity and \cref{fact:reduction-operator}}] \\
        &= \pseudred(x_{\verta, \verti}  m)\eqcomma &&\quad [\text{definition of $\pseudred$}]
    \end{align}
  \end{subequations}
  as desired. With Properties~\ref{item:rop-property-1}-\ref{item:rop-property-3} established, we conclude that $\pseudred$ is a degree\nobreakdash-$\degstd$ pseudo-reduction operator for $\polyhomo{\rela}{\relt}$.
\end{proof}

\bibliographystyle{alpha}

{\small
\bibliography{refs-local, refArticles, refBooks, refLocal
}
}
\end{document}